\def\year{2022}\relax
\documentclass[letterpaper]{article} 
\usepackage{aaai22}  
\usepackage{times}  
\usepackage{helvet}  
\usepackage{courier}  
\usepackage[hyphens]{url}  
\usepackage{graphicx} 
\usepackage{natbib}  
\usepackage{caption} 
\DeclareCaptionStyle{ruled}{labelfont=normalfont,labelsep=colon,strut=off} 
\usepackage{algorithm}
\usepackage{algorithmic}
\usepackage{amsmath}
\usepackage{amssymb}
\usepackage{amsthm}
\usepackage{multirow}
\usepackage{bm}
\usepackage{braket}
\usepackage{tikz}
\usetikzlibrary{calc}
\usetikzlibrary{patterns}
\usetikzlibrary{arrows.meta}
\usepackage{graphicx}

\usepackage[T3,T1]{fontenc}
\DeclareSymbolFont{tipa}{T3}{cmr}{m}{n}
\DeclareMathAccent{\invbreve}{\mathalpha}{tipa}{16}

\newtheorem{definition}{Definition}
\newtheorem{theorem}{Theorem}
\newcommand{\connected}{\rightleftharpoons}
\newcommand{\relevant}{\bowtie}
\newcommand{\game}{\mathcal{G}}
\newcommand{\blueprint}{\bm{{\xi}}_0}
\newcommand{\blueprintstrat}{{\xi}_0}
\newcommand{\polytopesub}{\Xi_j}
\newcommand{\planref}{\tilde{\bm{\xi}}} 
\newcommand{\plansub}{\tilde{\bm{\xi}}_j}
\newcommand{\plansubstrat}{\tilde{\xi}_j}
\newcommand{\sequencepre}{\hat{\Sigma}_i}

\usepackage{newfloat}
\usepackage{listings}
\floatstyle{ruled}
\newfloat{listing}{tb}{lst}{}
\floatname{listing}{Listing}
\title{Safe Subgame Resolving for Extensive Form Correlated Equilibrium}
\author{
    Chun Kai Ling,\\
    Fei Fang
}

\title{Safe Subgame Resolving for Extensive Form Correlated Equilibrium}
\author {
    Chun Kai Ling,
    Fei Fang
}
\affiliations {
    Carnegie Mellon University\\
    chunkail@cs.cmu.edu, feif@cs.cmu.edu
}

\usepackage{bibentry}

\begin{document}

\maketitle

\begin{abstract}
Correlated Equilibrium is a solution concept that is more general than Nash Equilibrium (NE) and can lead to outcomes with better social welfare. However, its natural extension to the sequential setting, the \textit{Extensive Form Correlated Equilibrium} (EFCE), requires a quadratic amount of space to solve, even in restricted settings without randomness in nature. To alleviate these concerns, we apply \textit{subgame resolving}, a technique extremely successful in finding NE in zero-sum games to solving general-sum EFCEs. Subgame resolving refines a correlation plan in an \textit{online} manner: instead of solving for the full game upfront, it only solves for strategies in subgames that are reached in actual play, resulting in significant computational gains. In this paper, we (i) lay out the foundations to quantify the quality of a refined strategy, in terms of the \textit{social welfare} and \textit{exploitability} of correlation plans, (ii) show that EFCEs possess a sufficient amount of independence between subgames to perform resolving efficiently, and (iii) provide two algorithms for resolving, one using linear programming and the other based on regret minimization. Both methods guarantee \textit{safety}, i.e., they will never be counterproductive. Our methods are the first time an online method has been applied to the correlated, general-sum setting.
\end{abstract}

\section{Introduction}
Correlation between players is a powerful tool in game theory. The \textit{Correlated Equilibrium} (CE) is an equilibrium that allows for players to coordinate actions with the aid of a mediator or a randomized correlation device, and is known to allow for outcomes which lead to a significantly higher social welfare as compared to solution concepts which require independent play, such as Nash Equilibrium (NE), on top of being computationally more tractable. 
In a CE, the mediator recommend actions privately to the players according to a probability distribution over joint actions that is known to all players, and the players have no incentive to deviate from the recommended action if they are perfectly rational.
A natural extension of CE to \textit{extensive form games} (EFG) is the \textit{Extensive Form Correlated Equilibrium} (EFCE), where players are recommended actions at each decision point~\cite{von2008extensive}. 
Unfortunately, solving and storing an EFCE typically requires space that is quadratic in the size of the game tree. This is a significant barrier towards solving large games: for example, storing an EFCE for a game of Battleship \cite{farina2019correlation} with a grid size of $3 \times 2$ requires a vector with more than $10^8$ entries. 

Over the last decade, a technique known as subgame resolving has gathered much attention amongst those looking to solve large games. The idea behind subgame resolving is to adopt a simple blueprint strategy at the beginning, and to compute refinements of the strategy in an \textit{online} manner only when the game has entered a \textit{subgame}. This means that one need not compute strategies in branches of the game which were never reached in actual play, just like with limited-depth search in perfect information games like chess. Rising into prominence because of the successes of superhuman-level poker bots such as \textit{Libratus} \citep{brown2017safe,brown2018superhuman}, subgame resolving has since been studied from other angles \cite{zhang2021subgame}, extended to other equilibrium concepts \citep{ling2021safe} and applied in practice to, multiplayer games like Hanabi \cite{lerer2020improving} and Diplomacy \citep{gray2020human}. However, subgame resolving has primarily been applied to the zero-sum or cooperative settings, with few inroads in the correlated setting, where the objective is to get players to coordinate despite potentially having misaligned interests.

In this paper, we introduce subgame resolving for EFCE. Instead of announcing the full correlation plan that specifies the probability of recommending different actions at each decision point, the mediator computes the EFCE strategies online. Conceptually, it can be viewed as having the mediator publish the algorithm of choosing recommended actions, and the algorithm is designed in a way such that the rational players will have no incentive to deviate from the recommended actions. Our contributions are twofold. First, we lay out the framework for safe subgame resolving for EFCE in terms of the exploitability of a correlation plan with respect to a correlation blueprint. Second, we show that for games without chance, the structure of the polytope of correlation plans contains a sufficient level of independence between subgames to facilitate independent solving. Third, we provide two refinement algorithms, the first based on a modification of the linear program (LP) of \citet{von2008extensive}, and the second utilizing a recent and more efficient method based on regret minimization \cite{farina2019efficient}. To the best of our knowledge, this is the first instance of subgame resolving being applied to the correlated setting. We experimentally show its scalability in benchmark games.

\section{Background and Related Work}
Let $\game$ be a 2-player extensive-form-game \textit{without chance}. This is represented by a finite game tree: nodes represent game states, belonging to either player $P_1$ or $P_2$, while actions are represented by edges directed down the tree. 
To represent imperfect information, $\game$ is supplemented with \textit{information sets} (infosets) $I_i \in \mathcal{I}_i$, $i \in [2]$, which are collection of states belonging to but are indistinguishable to $P_i$. 
States in the same infoset contain the same actions $a_i \in \mathcal{A}(I_i)$. 
We denote by $ha$ the state that is reached immediately after taking action $a$ at state $h$.
We say that state $h$ precedes ($\sqsubset$) $h'$ if $h\neq h'$ and $h'$ is a descendent of $h$ in the game tree, and use the notation $h \sqsubseteq h'$ when allowing $h=h'$. 
We assume players have perfect recall, that is, players never forget past observations and past actions.
The set of terminal states $\mathcal{L}$ are known as \textit{leaves}. Each leaf is associated with utilities received by each player $u_i(h)$. For a given leaf $h$, the \textit{social welfare} is given by $u_1(h) + u_2(h)$.

We define the set of \textit{sequences} for $P_i$ as the set $\Sigma_i := \{ (I, a) : I \in \mathcal{I}_i, a \in \mathcal{A}(I) \} \cup \{ \varnothing \}$, where $\varnothing$ is known as the \textit{empty sequence}. For any infoset $I_i \in \mathcal{I}_i$, we denote by $\sigma(I)$ the \textit{parent sequence} of $I$, which is defined as the (unique) sequence which precedes $I$ from the root to any node in $I$; if no such sequence exists, then $\sigma(I) = \varnothing$. Sequences in $\Sigma_i$ form a partial order; for sequences $\tau=(I,a), \tau'=(I',a') \in \Sigma_i$, we write $\tau \prec \tau'$ if there exists states $ha$, $h' \in I'$ belonging to $P_i$ such that $ha \sqsubseteq h'$, and write $\tau \preceq \tau'$ if allowing $\tau=\tau'$. If in addition, $\sigma(I') = \tau$, we say that $\tau'$ is an immediate successor of $\tau$ and write $\tau \prec_1 \tau'$. Since the game has no chance, each leaf $h \in \mathcal{L}$ is uniquely identified by a pair of sequences $(\sigma_1, \sigma_2)$. With a slight abuse of notation we write $(\sigma_1, \sigma_2) \in \mathcal{L}$, and denote corresponding player payoffs and social welfare by $u_i(\sigma_1, \sigma_2)$ and $u(\sigma_1, \sigma_2)$.

\paragraph{Sequence-form strategies} In the sequence form, a (mixed) strategy for $P_i$ is compactly represented by a vector $x_i$, indexed by the sequences $\sigma=(I,a) \in \Sigma_i$. The entry $x_i[\sigma]$ contains the \textit{product} of the probabilities of $P_i$ taking actions from the root to information set $I$\footnote{Perfect recall means there is only one such series of actions.}, including $a$ itself, with the base case given by $x_i[\varnothing]=1$. Hence, valid sequence-form strategies must satisfy the `flow' constraints; for every $I \in \mathcal{I}_i$, we have $\sum_{a \in \mathcal{A}(I)} x_i[(I,a)] = \sigma(I)$. Sequence-form strategies have size roughly equal to the number of actions of the player, while flow constraints can be seen as a generalization of the sum-to-one constraints for strategies in the simplex.

\subsection{Extensive-Form Correlated Equilibria} Extensive-form correlated equilibria (EFCE) is a natural extension of CE to EFGs. Unlike regular CEs, players do not receive recommendations for the full game upfront; instead, recommendations are received sequentially, and only for infosets the players are currently in. In the original paper by \citet{von2008extensive}, this is achieved by means of \textit{sealed recommendations}, while \citet{farina2019correlation} have the mediator generating recommendations over the course of the game, but ceasing all future recommendations if a player deviates from a recommendation. We call the recommended actions \textit{trigger sequences} $\sigma^!$ \citep{dudik2012sampling}. Trigger sequences contain the last recommended action from the mediator before any deviation, and implicitly contains information about all previous recommendations (due to perfect recall).
EFCEs are \textit{incentive-compatible}, players do not expect to benefit by unilaterally deviating.

\paragraph{Polytope of Correlation Plans} A significant benefit of EFCEs over regular CEs is computational cost: computing a CE that achieves maximum social welfare is NP-complete \cite{von2008extensive}, while in 2-player perfect recall games without chance\footnote{and more generally in games that are triangle-free\cite{farina2020polynomial}}, the constraints that define an EFCE may be expressed in a polynomial number of linear constraints and hence may be solved using a linear program. Crucial to these positive results is a theorem by \citeauthor{von2008extensive} which characterizes $\Xi$, the \textit{polytope of correlation plans} which compactly represents the space of joint (reduced) normal-form strategies up to strategic equivalence. 

\begin{definition}{\textup{(Connected infosets, $I \connected I'$)}} 
Let $I, I'$ be infosets from either player. We say that $I, I'$ are \textit{connected} and write $I \connected I'$ if there exists nodes $u \in I, v \in I'$ in $\game$  lying on a path starting from the root.
\end{definition}

\begin{definition}{\textup{(Relevant sequences, $\sigma_1 \relevant \sigma_2$)}}
Let $\sigma_1 \in \Sigma_1, \sigma_2 \in \Sigma_2$. We say that the sequence pair $(\sigma_1, \sigma_2)$ is relevant, denoted by $\sigma_1 \relevant \sigma_2$ if (i) either $\sigma_1$ or $\sigma_2$ is $\varnothing$ or (ii) $\sigma_1 = (I_1, a_1), \sigma_2 = (I_2, a_2)$ for $I_1 \connected I_2$ and some actions $a_1, a_2$. For convenience, we use the same notation $\sigma_1 \relevant I_2$ when either $\sigma_1 = \varnothing$ or if $\sigma_1 = (I_1, a_1)$ and $I_1 \connected I_2$, with a symmetric definition for $I_1 \connected \sigma_2$.
\end{definition}

\begin{definition}{\textup{(\citeauthor{von2008extensive})}}
Let $\game$ be a perfect recall game without chance. Then, $\Xi$ is a convex polytope of correlation plans which contains non-negative vectors indexed by relevant sequence pairs, with constraints
\begin{align*}
    \Xi := 
    \left\{
    \bm{\xi} \geq 0: 
    \begin{array}{ccc}
    \bm{\xi} [\varnothing, \varnothing] = 1, \\ 
    \sum_{a \in \mathcal{A}(I)} \xi[(I_1,a), \sigma_2] = \xi[\sigma(I_1), \sigma_2], \\
    \sum_{a \in \mathcal{A}(I)} \xi[\sigma_1, (I_2, a)] = \xi[\sigma_1, \sigma(I_2)]
    \end{array}
    \right\}
\end{align*}
where the second (and third) constraint is over all $I_1 \relevant \sigma_2$ ($\sigma_1 \connected I_2$).
\label{def:xi_polytope}
\end{definition}
Visually, one can view $\Xi$ as a 2-dimensional `checkerboard' of size $|\Sigma_1| \cdot |\Sigma_2|$ with entries to be filled in indices where $\sigma_1 \relevant \sigma_2$. The second and third constraints are simply the sequence-form constraints \cite{von2008extensive} applied to each row and column of the checkerboard. For example, for the game in Figure~\ref{fig:signalling}, all sequence pairs are relevant, and we have row constraints $\xi[\sigma_1, \ell_x] + \xi[\sigma_1, r_x] = \xi[\sigma_1, \varnothing]$ and  $\xi[\sigma_1, \ell_y] + \xi[\sigma_1, r_y] = \xi[\sigma_1, \varnothing]$ for all sequences $\sigma_1 \in \Sigma_1$, and column constraints $\xi[G, \sigma_2] + \xi[B, \sigma_2] = \xi[\varnothing, \sigma_2]$,
$\xi[X_G, \sigma_2] + \xi[Y_G, \sigma_2] = \xi[G, \sigma_2]$, and 
$\xi[X_B, \sigma_2] + \xi[Y_B, \sigma_2] = \xi[B, \sigma_2]$ for all $\sigma_2 \in \Sigma_2$. 
\begin{figure}
\begin{minipage}{0.55\linewidth}
\begin{tikzpicture}
\tikzstyle{solid node}=[circle,draw,inner sep=1.2,fill=black];
\tikzstyle{none}=[];
		\node [label=left:$1$, style=solid node] (0) at (-11, 7) {};
		\node [label=left:$1$, style=solid node] (1) at (-11.75, 5.75) {};
		\node [label=left:$1$, style=solid node] (2) at (-10.25, 5.75) {};
		\node [style=solid node] (3) at (-12.75, 4.25) {};
		\node [style=solid node] (4) at (-10.5, 4.25) {};
		\node [style=solid node] (5) at (-11.5, 4.25) {};
		\node [style=solid node] (6) at (-9.25, 4.25) {};
		\node [style=none] (9) at (-13, 3.) {};
		\node [style=none] (10) at (-12.5, 3.) {};
		\node [style=none] (11) at (-11.75, 3.) {};
		\node [style=none] (12) at (-11.25, 3.) {};
		\node [style=none] (13) at (-10.75, 3.) {};
		\node [style=none] (14) at (-10.25, 3.) {};
		\node [style=none] (15) at (-9.5, 3.) {};
		\node [style=none] (16) at (-9, 3.) {};
		\node [style=none] (23) at (-11.75, 3.) {};
		\draw (0.center) -- (1.center) node [midway, left] {$G$};
		\draw (0.center) -- (2.center) node [midway, right] {$B$};
		\draw (1.center) -- (3.center) node [midway, left] {$X_G$};
		\draw (1.center) -- (4.center) node [midway, left] {$Y_G$};
		\draw (2.center) -- (5.center) node [midway, right] {$X_B$};
		\draw (2.center) -- (6.center) node [midway, right] {$Y_B$};
		\draw (3.center) -- (9.center) node [below=-0.8mm] (lxL) {\small $\ell_X$};
		\draw (3.center) -- (10.center) node [below] {\small $r_X$};
		\draw (5.center) -- (11.center) node [below=-0.8mm] {\small $\ell_X$};
		\draw (5.center) -- (12.center) node (rxR) [below] {\small $r_X$};
		\draw (4.center) -- (13.center) node (lyL)[below=-0.8mm] {\small $\ell_Y$};
		\draw (4.center) -- (14.center) node [below] {\small $r_Y$};
		\draw (6.center) -- (15.center) node [below=-0.8mm] {\small $\ell_Y$};
		\draw (6.center) -- (16.center) node (ryR) [below] {\small $r_Y$};

\draw[rounded corners=7]
($(0)+(-0.5,.25)$)rectangle($(0)+(.5,-0.25)$);

\draw[rounded corners=7]
($(1)+(-0.5,.25)$)rectangle($(1)+(.5,-0.25)$);

\draw[rounded corners=7]
($(2)+(-0.5,.25)$)rectangle($(2)+(.5,-0.25)$);

\draw[rounded corners=7]
($(3)+(-0.25,.25)$)rectangle($(5)+(.25,-0.25)$);
\node (2a) at($(3)!0.5!(5)$){2};

\draw[rounded corners=7]
($(4)+(-0.25,.25)$)rectangle($(6)+(.25,-0.25)$);
\node (2b) at($(4)!0.5!(6)$){2};

\draw[dashed, fill=gray,opacity=0.4]
($(lxL)+(-0.2, -0.2)$) rectangle ($(rxR)+(0.2, 1.6)$);

\draw[dashed, fill=gray, opacity=0.7]
($(lyL)+(-0.2, -0.2)$) rectangle ($(ryR)+(0.2, 1.6)$);

\end{tikzpicture}
\end{minipage}
\begin{minipage}{0.40\linewidth}
\begin{tikzpicture}
\small
\tikzstyle{circle node}=[circle,draw,inner sep=0.0,fill=white,text width=3mm, align=center]

\node [] (1) at (-0.5, 0) {$\varnothing$};
\node [] (1) at (-0.5, -0.6) {$G$};
\node [] (1) at (-0.5, -1.1) {$B$};
\node [] (1) at (-0.5, -1.7) {$X_G$};
\node [] (1) at (-0.5, -2.2) {$Y_G$};
\node [] (1) at (-0.5, -2.8) {$X_B$};
\node [] (1) at (-0.5, -3.3) {$Y_B$};

\node [] (1) at (0, 0.5) {$\varnothing$};
\node [] (1) at (0.6, 0.5) {$\ell_x$};
\node [] (1) at (1.1, 0.5) {$r_x$};
\node [] (1) at (1.7, 0.5) {$\ell_y$};
\node [] (1) at (2.2, 0.5) {$r_y$};

\draw[]
(-.25, -.25)rectangle(0.25, 0.25) node [pos=0.5,style=circle node] {1};
\draw[pattern=north west lines]
(0.35, -.25)rectangle(1.35, 0.25) node [pos=0.5,style=circle node] {11};
\draw[pattern=north east  lines]
(1.45, -.25)rectangle(2.45, 0.25)node [pos=0.5,style=circle node] {12};

\draw[]
(-.25, -1.35)rectangle(0.25, -0.35) node [pos=0.5,style=circle node] {2};
\draw[pattern=north west lines]
(0.35, -1.35)rectangle(1.35, -0.35) node [pos=0.5,style=circle node] {9};
\draw[pattern=north east lines]
(1.45, -1.35)rectangle(2.45, -0.35) node [pos=0.5,style=circle node] {10};

\draw[]
(-.25, -2.45)rectangle(0.25, -1.45) node [pos=0.5,style=circle node] {3};
\draw[pattern=north west lines]
(0.35, -2.45)rectangle(1.35, -1.45) node [pos=0.5,style=circle node] {4};
\draw[pattern=north east lines]
(1.45, -2.45)rectangle(2.45, -1.45)node [pos=0.5,style=circle node] {5};

\draw[]
(-.25, -3.55)rectangle(0.25, -2.55) node [pos=0.5,style=circle node] {6};
\draw[pattern=north west lines]
(0.35, -3.55)rectangle(1.35, -2.55) node [pos=0.5,style=circle node] {7};
\draw[pattern=north east lines]
(1.45, -3.55)rectangle(2.45, -2.55) node [pos=0.5,style=circle node] {8};





\end{tikzpicture}
\end{minipage}
\caption{Left: Modified signaling game used in \citep{von2008extensive} with 2 subgames. Right: Correlation plan $\xi$. Circles denote fill-in order under the decomposition of \citet{farina2019efficient}. Dashed rectangles show sequence pairs in different subgames.}
\label{fig:signalling}
\end{figure}
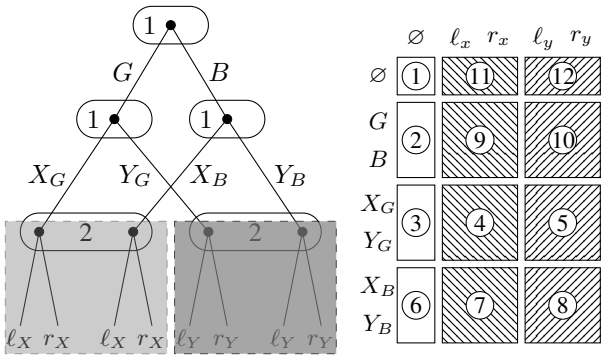
\paragraph{LP-based EFCE solvers.} Observe that $\Xi$ contains a polynomial number of unknowns and linear constraints. A correlation plan in $\Xi$ is a EFCE if it also satisfies incentive constraints that enforce incentive compatibility such that it is optimal for a player to follow the recommendation. \citeauthor{von2008extensive} show that the incentive constraints can be also expressed in a polynomial number of linear constraints over $\bm{\xi}$. Specifically, 
incentive constraints when $\sigma^!=(I,a^!)$ is recommended for $P_1$ (the case for $P_2$ follows naturally) are expressed by\footnote{Readers familiar with the work of \citet{von2008extensive} will notice that we use a slightly different LP. This is to make our future definition of exploitability more convenient.}
{\small
\begin{align}
    \mu(\sigma^!) &\geq 
    \beta(\sigma'; \sigma^!) \qquad
    \sigma' = (I, a'), a' \in \mathcal{A}(I) \backslash \{ a^! \}
    \label{eq:lp_incentive}
    \\
    \mu(\sigma) &= \sum_{\sigma_2; (\sigma, \sigma_2) \in \mathcal{L}} u_1(\sigma, \sigma_2) \xi[\sigma, \sigma_2] + \sum_{\sigma' \succ_1 \sigma} \mu(\sigma') 
    \label{eq:lp_mu}
    \\
    \beta(\sigma_1; \sigma^!) &= 
    \sum_{
    \substack{(\sigma_1, \sigma_2) \\ \in \mathcal{L}}
    }
    u_1(\sigma_1, \sigma_2) \xi[\sigma^!, \sigma_2] + 
     \sum_{\substack{I'; \sigma(I')\\ =\sigma_1}} \nu(I'; \sigma^!)
     \label{eq:lp_beta}
    \\
    \nu(I; \sigma^!) &\geq 
    \beta(\sigma; \sigma^!) \qquad a \in \mathcal{I}(I)
    \label{eq:lp_nu}
\end{align}
}
Here, $\mu(\sigma)$ gives the expected utility of $P_1$ if he abides to this and all following recommendations. 
Together, \eqref{eq:lp_beta} and \eqref{eq:lp_nu} recursively define the values of the best response of $P_1$ for deviating to $\sigma'$ given $\sigma^!$ was recommended. The term $\xi[\sigma^!, \sigma_2]$ essentially contains the (unnormalized) posterior of $P_2$'s sequence given that $\sigma^!$ was recommended.
\paragraph{Bilinear Saddle-point Problems and Regret Minimization} More recent work by \cite{farina2019correlation,farina2019efficient} show that the problem of finding an EFCE can be formulated as a bilinear saddle point problem, i.e., an optimization problem of the form $\min_{\textbf{x} \in \mathcal{X}} \max_{\textbf{y} \in \mathcal{Y}} \textbf{x}^TA\textbf{y}$. Conceptually, this can be seen a \textit{zero-sum} game between two entities, (i) a \textit{mediator}, who optimizes $\bm{\xi} \in \Xi$, and (ii) a \textit{deviator}, who selects, for each sequence $\sigma^! \in \Sigma_i$, the strategy (for all $\sigma \succ \sigma^!$) that is to be taken after deviating from $\sigma^!$, given the mediator's choice of $\bm{\xi}$. Essentially, the mediator tries to increase the value of $\mu$, while the deviator seeks to increase $\nu$, which makes the inequality in \eqref{eq:lp_incentive} more difficult to achieve. 
\citeauthor{farina2019efficient} characterize $\Xi$ in terms of a series of convexity-preserving operations known as \textit{scaled extensions} and provide a regret minimizer for sets constructed via scaled extensions. 
This construction leads to an efficient EFCE solver that runs in linear space, which we adapt in one of our resolving algorithms. 
\paragraph{Quality of correlation plans}The quality of any correlation plan $\bm{\xi}$ is measured by (i) its expected \textit{social welfare}, $\sum_{(\sigma_1, \sigma_2) \in \mathcal{L}} \xi[\sigma_1, \sigma_2] u(\sigma_1, \sigma_2)$, where $u$ is typically the payoff sum $u_1+u_2$, and (ii) the degree to which the $\xi$ violates the incentive constraints. 
\begin{definition}{\textup{(Exploitability)}}
Given a trigger sequence $\sigma^!$ of $P_1$, and a strategy $\bm{\xi} \in \Xi$, the value of the best-deviating response to $\sigma^!=(I, a^!)$ is given by 
{\small
\begin{align*}
\beta^*(\bm{\xi};\sigma^!) &= \max_{a \in \mathcal{I} \setminus \{ a^! \}} \beta((I, a), \bm{\xi};\sigma^!) \\ 
\beta(\sigma, \bm{\xi};\sigma^!) &= \sum_{
\substack{\sigma_2:(\sigma, \sigma_2) \in \mathcal{L}}} u_1(\sigma, \sigma_2) \xi[\sigma^!, \sigma_2] +  \sum_{I;\sigma(I)=\sigma} \nu(I, \bm{\xi}; \sigma^!)
\\
\nu(I, \bm{\xi};\sigma^!) &= \max_{a\in\mathcal{I}\{a'\}} \beta((I,a), \bm{\xi}; \sigma^!)
\end{align*}
}
with a similar definition for $P_2$. The exploitability of $\bm{\xi}$ for a trigger sequence $\sigma^!$ is given by 
\begin{align*}
    \delta^*(\bm{\xi}, \sigma^!) = \beta^*(\bm{\xi}; \sigma^!) - \mu(\bm{\xi}, \sigma^!)
\end{align*}
where $\mu(\bm{\xi}, \sigma^!)$ is the value of the $\sigma^!$ if it and all future recommendations are followed, as defined in \eqref{eq:lp_mu}.
\label{def:exploitability}
\end{definition}
$\beta^*(\bm{\xi};\sigma^!)$ is the highest reward a player can get from deviating from the trigger sequence $\sigma^!$, while $\delta^*$ measures the gain from doing so. If $\delta^*(\bm{\xi}; \sigma^!) \leq 0$ for all $\sigma^! \in \mathcal{I}_1 \cup \mathcal{I}_2$, then $\bm{\xi}$ is an EFCE. In LP-based solvers, the social welfare is maximized through the objective function, while exploitability is $\leq 0$ using linear constraints. In the regret minimization method, exploitability is bounded by the average regret incurred by the solvers, which goes to $0$ at a rate of $1/\sqrt{T}$. Maximizing social welfare with regret minimizers typically requires performing binary search.

\subsection{Subgames}
An EFG's tree structure provides a natural means of decomposing the problem of solving a game into smaller subproblems over subtrees. However, in imperfect information games, we will require additional restrictions..
\begin{definition}{\textup{(Subgame)}} Let $\game$ be an EFG with perfect recall. 
Let $H$ be a subset of nodes in $\game$ and $\check{\game}_H$ be the subgraph induced by $H$. We call $\check{\game}_H$ a subgame of $\game$ when:
(i) if state $s\in H$, then $s' \sqsupset s$ implies $s' \in H$, and (ii) for all information sets $I \in \mathcal{I}_1 \cup \mathcal{I}_2$, we have $H \cap I= I$ or $\emptyset$.\footnote{Alternatively if $h \in \check{\game}$ and belongs to some infoset $I$, then all states $h' \in I$ are contained in $\check{\game}$.} 
\label{def:subgame}
\end{definition}
\begin{definition}{\textup{(Subgame Decomposition)}}
Let $\mathcal{H} = \{H_j\}$ be sets of vertices of $\game$. We call $\mathcal{H}$ a valid subgame decomposition if (i) $\mathcal{H}$ contains non-intersecting sets, (ii) each $H_j \in \mathcal{H}$ induces a valid subgame $\check{\game}_{H_j}$ ($\check{\game}_j$ for short).
\end{definition}
For this paper, we will assume that we are equipped with a valid subgame decomposition $\mathcal{H}$, which induces $J$ disjoint subgames $\{ \check{\game}_j \}$. There are many possible ways to obtain subgame decomposition, but by far the most natural and common one is based on \textit{public information}. In this paper, we make no additional assumptions on subgames apart from those in the definition. We call nodes that are not included in any $\mathcal{H}$ as \textit{pre-subgame}, with an induced subtree $\hat{\game}$. Note that $\hat{\game}$ obeys property (ii) of a subgame; if some infoset is only partially contained in $\hat{\game}$, then it must be partially contained in some subgame, which is disallowed. Consequently, leaves, infosets, and sequences may be likewise partitioned. We denote these sets by $\check{\mathcal{L}}_j, \hat{\mathcal{L}}$,  $\check{\mathcal{I}}_{i,j}, \hat{\mathcal{I}}_{i}$, and $\check{\Sigma}_{i,j}, \sequencepre$. 

The game in Figure~\ref{fig:signalling} has two subgames, both starting off with $P_2$ making his move. Here, $P_2$'s infosets belong to separate subgames, while $P_1$'s infosets all lie in $\hat{\game}$. Similarly, all of $P_2$'s non-empty sequences lie in a subgame, while all of $P_1$'s sequences do not. Another valid subgame decomposition is to have all but the root be in a single subgame. 

\section{Subgame Resolving for EFCE}
Subgame resolving exploits the sequential nature of EFGs to refine strategies online. We begin with a \textit{correlation blueprint}, typically a guess or approximate of an EFCE. 
\begin{definition}{\textup{(Correlation blueprint)}}
A correlation blueprint $\blueprint \in \Xi$ for the game $\game$ is an oracle $\blueprintstrat[\sigma_1, \sigma_2]$ which can be accessed in constant time for all $\sigma_1 \relevant \sigma_2$.
\end{definition}
Note that blueprint strategies $\blueprintstrat$ may not necessarily be stored explicitly: all we require is that its entries may be accessed efficiently. For example, a blueprint may have players playing independently according to sequence form strategies $\xi^{(i)}(\sigma)$, such that $\blueprintstrat[\sigma_1, \sigma_2] = \xi^{(1)}(\sigma_1) \cdot \xi^{(2)}(\sigma_2)$ (no correlation between players' actions in this special blueprint). 

At the beginning of the game, players receive recommendations from the blueprint strategy. Once the game enters a subgame, an equilibrium refinement step is performed \textit{only for that subgame entered}, and recommended actions are instead drawn from that refined correlation plan for the rest of the game. Subgame resolving is an \textit{online} method; instead of solving for the equilibrium upfront, it defers part of its computation to when the game is being played. A generic algorithm is shown in Algorithm~\ref{alg:generic_sr}. 

\paragraph{Refinements of correlation blueprint.} Subgame resolving for EFCEs differs significantly from prior work for zero-sum and Stackelberg games. This is because we are now updating relevant sequence pairs of $\blueprint$ in the correlation polytope $\Xi$, which unlike the space of sequence form strategies, has no obvious hierarchical structure. Fortunately, Definitions~\ref{def:xi_polytope} and \ref{def:subgame} provide enough structure to perform resolving.
\begin{theorem}{\textup{(Independence between subgames)}}
Let the set $S_j$ contain relevant sequences (i) $\sigma_1, \sigma_2 \in \check{\game}_j$, or (ii) $\sigma_1 \in \hat{\game}, \sigma_2 \in \check{\game}_j$, or (iii) $\sigma_1 \in \check{\game}_j, \sigma_2 \in \hat{\game}$. Let $S_0$ be the set of relevant sequence pairs such that $\sigma_1, \sigma_2 \in \hat{\game}$. Then $\{ S_0,\cdots,S_J \}$ forms a partition of relevant sequence pairs.
\label{thm:independence_subgames}
\end{theorem}
A relevant sequence pair $(\sigma_1, \sigma_2)$ is \textit{pre-subgame}, written $(\sigma_1, \sigma_2) \in \hat{\game}$ if $(\sigma_1, \sigma_2) \in S_0$ and $(\sigma_1, \sigma_2) \in \check{\game}_j$ if $(\sigma_1, \sigma_2) \in S_j$. Theorem~\ref{thm:independence_subgames} shows exactly one of these must hold.
\begin{definition}{\textup{(Refinements)}}
For a given blueprint $\blueprint \in \Xi$ and subgame decomposition $\mathcal{H}$, a correlation plan $\tilde{\bm{\xi}} \in \Xi$ is called a complete refinement if $\tilde{\xi}[\sigma_1, \sigma_2] = \blueprintstrat[\sigma_1, \sigma_2]$ for all $(\sigma_1, \sigma_2) \in \hat{\game}$.
Let $\polytopesub$ be $\Xi$ but restricted to sequence pairs $(\sigma_1, \sigma_2) \in \check{\game}_j \cup \hat{\game}$. 
We call $\plansub \in \polytopesub$ a refinement of subgame $j$ if $\plansubstrat[\sigma_1, \sigma_2] = \blueprintstrat[\sigma_1, \sigma_2]$ for all $(\sigma_1, \sigma_2) \in \hat{\game}$. 
\end{definition}
For example, in Figure~\ref{fig:signalling}, a complete refinement involves updating all but the first column, since for $P_2$, all but the empty sequence is in some subgame. For the left subgame, we have $\polytopesub$ being the first 3 columns; finding a refinement involves updating the columns containing $\ell_x, r_x$ (sequences which are contained in the subgame) and dropping the last $2$ columns, while respecting the constraints in Definition~\ref{def:xi_polytope}. 

Theorem~\ref{thm:independence_subgames} implies that updated entries (shaded entries) for each refinement do not overlap, hence, refined correlation plans can be combined to form complete refinements. Let $\{ \plansub \}$ contain a refinement for each subgame. Then, $\{ \plansub \}$ \textit{induces} a complete refinement naturally, $\tilde{\xi}[\sigma_1, \sigma_2] = \blueprintstrat[\sigma_1, \sigma_2]$ if $(\sigma_1, \sigma_2) \in \hat{\game}$ and $\plansubstrat[\sigma_1, \sigma_2]$ if $(\sigma_1, \sigma_2) \in \check{\game}_j$. This direct mapping satisfies $\tilde{\bm{\xi}} \in \Xi$, as no constraint of $\Xi$ involves sequences pairs belonging to  different subgames. 

The independence property of sequence pairs extends to EFCE incentive constraints for trigger sequences within subgames. For every trigger sequence $\sigma^!$ in $\check{\game}_j$, the best-deviating response (see Definition~\ref{def:exploitability}) will never have to reference sequence pairs containing any sequence outside $\check{\game}_j$.\footnote{This is intuitively true. Once inside $\check{\game}_j$, a potential deviating player will never encounter states outside of $\check{\game}_j$ in the future, and hence need not consider them.} These show it may be possible to perform refinements of subgames \textit{independently} without solving other entries containing sequences from other subgames. However, independence of incentive constraints does not apply to pre-subgame trigger sequences $\sigma^! \in \sequencepre$. For those sequences, $\delta^*(\bm{\xi}, \sigma^!)$ will in general depend on refined solutions from multiple, distinct subgames. Handling these constraints is a primary challenge addressed in this paper.

\begin{algorithm}[tb]
\caption{Subgame Resolving}
\label{alg:generic_sr}
\textbf{Input}: EFG, blueprint $\blueprint$
\begin{algorithmic}[1] 
\WHILE{game is not over}
\IF {currently in some subgame $j$}
\IF {first time in subgame }
    \STATE (*) Refine $\blueprint \rightarrow \plansub$
\ENDIF
\STATE Recommend action according to $\plansub$
\ELSE
\STATE Recommend action according to $\blueprint$
\ENDIF
\ENDWHILE
\end{algorithmic}
\end{algorithm}
\paragraph{Safe refining algorithms} 
An important property when performing subgame-resolving for independent, uncorrelated strategies is that of safety, and was the central issue discussed extensively in solving NE in zero-sum games~\cite{brown2017safe}. There, it was observed naive application of resolving algorithms can result in solutions which are of lower quality than the blueprint. The fundamental problem is that when $P_1$ performed resolving, the best-response of $P_2$ in the pre-subgame portion differs from the blueprint, hence whatever initial distribution over states at the beginning of the subgame no longer holds. This phenomenon is known as \textit{unsafe} resolving. A similar phenomenon quantified in terms of exploitability holds for EFCEs.  
\begin{definition}{\textup{(Safe refinements)}}
A complete refinement $\tilde{\bm{\xi}}$ of $\blueprint$ is safe if for all trigger sequences $\sigma^!$
\begin{align*}
    \delta^* (\tilde{\bm{\xi}}; \sigma^!)  \leq \max \left(0,  \delta^* 
    (\blueprint; \sigma^!) \right),
\end{align*}
i.e., the exploitability of $\tilde{\bm{\xi}}$ for all $\sigma^!$ is $0$ or less than the blueprint. We say that $\tilde{\bm{\xi}}$ is fully safe if in addition, the social welfare (assuming no deviations) under $\tilde{\bm{\xi}}$ is no less than $\blueprint$. A resolving algorithm is said to be (fully) safe if the complete refinement induced by all $j$ refinements $\plansub$ is (fully) safe.
\label{def:safe_resolving}
\end{definition}
In safe refinements, players are at least as incentivized to follow the resolved strategy than the blueprint. Fully safe refinements ensures further that the social welfare will not be diminished. Apart from incurring additional computing costs, there can be no harm in applying fully safe resolving. Clearly, a fully safe resolving algorithm exists in the form of one that trivially returns the blueprint. 

\paragraph{Resolving with multiple subgames.} In Definition~\ref{def:safe_resolving}, we required that the induced complete refinement $\tilde{\bm{\xi}}$ be used to measure safety, and not just the refined strategy of a subgame $\plansub$. This may seem odd at first, since the primary advantages of resolving was that it did not require computing strategies for subgames not reached in actual play. However, it turns out that this is necessary. Consider the perspective of $P_i$ who in the pre-subgame portion of $\game$ was recommended a sequence $\sigma^!$ and was considering deviation. At that point of decision making, $P_i$ does not know which subgame will be reached in the future; however, he knows that whichever subgame is encountered (if at all), refinement will be performed. Thus, when contemplating deviation, $P_i$ in fact computes the value of a best-deviating response to the complete refinement $\tilde{\bm{\xi}}$. This is despite the fact that in a single playthrough of the game, at most one subgame can be encountered in reality. Another interpretation is that the mediator publishes the refinement \textit{algorithm} which implicitly defines the complete  $\tilde{\bm{\xi}}$, which players contemplate best responses to. Hence, even though the resolving algorithm does not explicitly compute a complete refinement, it should still guarantee safety \textit{as if} it did.

\section{Safe Subgame Resolving Using LPs} \label{sec:resolve_lp}
Suppose the mediator has thus far given recommendations based on $\blueprint$ and the players have just entered subgame $j$. Following Algorithm~\ref{alg:generic_sr}, the mediator computes a refinement $\plansubstrat$ which he uses for all future recommendations. We now present a safe refinement algorithm using a LP (a fully safe variant will be discussed later). 

On top of the structural constraints of $\polytopesub$, we have 3 categories (A-C) of additional constraints that ensure safety. Constraint set (A) enforces safety for trigger sequences $\sigma^! \in \check{\game}_j$, in an manner identical to \eqref{eq:lp_incentive}, while constraint sets (B-C) ensures that the complete refinement $\tilde{\bm{\xi}}$ is safe; loosely speaking, (B) contains lower bounds that ensure that following recommendations will yield a high enough payoff to a player contemplating deviation, while (C) contains upper bounds which ensure that players which have deviated do not get rewarded too much. 

\paragraph{(A) Safety for in-subgame triggers} 
For each $P_i$ and each sequence in subgame $j$, i.e., $\sigma^! = (I^!, a^!) \in \check{\Sigma}_{i,j}$, we require
{\small
\begin{align}
    \mu(\tilde{\bm{\xi}}; \sigma^!) \geq \beta^*(\tilde{\bm{\xi}}; \sigma^!) - \delta^*(\blueprint; \sigma^!) \label{eq:incentive_inside_subgame}
\end{align}
}
where the $\mu, \nu$ have constraints identical to \eqref{eq:lp_mu}, \eqref{eq:lp_nu}. These constraints only involve sequence (pairs) that lie within $\check{\game}_j$ and not other subgames, so no modifications are needed. 


\paragraph{Computing safe infoset value bounds} Now we turn to constraints (B) and (C), which guarantee safety for trigger sequences in $\hat{\game}$. Our approach is to, for each trigger-sequence $\sigma^!=(I,a^!)$,  generate a set of linear constraints which guarantee that the safety for $\sigma^!$ is satisfied, in accordance to Definition~\ref{def:safe_resolving}. 
What are some sufficient conditions on $\mu(\tilde{\bm{\xi}}; \sigma^!)$ and $\beta(\sigma', \tilde{\bm{\xi}}; \sigma^!)$ such that the safety condition in Definition~\ref{def:safe_resolving} is satisfied for $\sigma^!$?
To answer this, let us consider $\alpha=\max(0, \delta^*(\blueprint; \sigma^!))$. 
There are $2$ cases. (i) If $\alpha=0$, then the blueprint was already sufficiently unexploitable for $\sigma^!$. Thus we could afford to decrease $\mu(\tilde{\bm{\xi}}; \sigma^!)$ and increase $\beta(\sigma', \tilde{\bm{\xi}}; \sigma^!)$ relative to the blueprint---if it leads to better social welfare. (ii) If $\alpha \geq 0$, then $\sigma^!$ was exploitable and we do not want to worsen exploitability. This can be avoided if we could somehow ensure $\mu(\bm{\xi}, \sigma^!)$ and $\beta(\sigma', \tilde{\bm{\xi}}; \sigma^!)$ do not decrease or increase respectively. 
Concretely, in case (i), we can require $\beta^*(\tilde{\bm{\xi}}; \sigma^!) \leq \invbreve{\beta}(\sigma; \sigma^!) =  \beta^*(\blueprint; \sigma^!) - \delta^*(\blueprint; \sigma^!)/2$, and $\mu(\tilde{\bm{\xi}}; \sigma^!) \geq \breve{\mu}(\sigma; \sigma^!) = \mu(\blueprint, \sigma^!) + \delta^*(\blueprint; \sigma^!)/2$. 
In case (ii), we can require $\beta^*(\tilde{\bm{\xi}}; \sigma^!) \leq \invbreve{\beta}(\sigma; \sigma^!) =  \beta^*(\blueprint; \sigma^!)$, and $\mu(\tilde{\bm{\xi}}; \sigma^!) \geq \breve{\mu}(\sigma; \sigma^!) = \mu(\blueprint, \sigma^!)$. These are sufficient conditions to guarantee that safety is maintained for $\sigma^!$. Yet, enforcing this is not possible, since $\mu(\tilde{\bm{\xi}}; \sigma^!)$ and $\beta(\sigma', \tilde{\bm{\xi}}; \sigma^!)$ can depend on relevant sequence pairs belonging to other subgames. The trick is to recursively unroll $\mu$ and $\beta$, maintaining bounds which guarantee for safety at each step. This is repeated until we reach infosets belonging to subgames. 

\begin{definition}{\textup{(Head infosets)}} 
For a subgame $j$, $I \in \mathcal{I}_i$ is a head infoset of subgame $j$ if $I \in \hat{\mathcal{I}}_{i,j}$ and there does not exist $I' \prec I$ such that $I' \not\in \hat{\mathcal{I}}_i$. 
The set of head infosets for player $i$ in subgame $j$ is denoted by $\mathcal{I}_{i,j}^{\text{h}} \subseteq \check{\mathcal{I}}_{i,j}$. 
$I$ is called a head infoset if it is a head infoset of some subgame.
\end{definition}
\paragraph{(B) Lower bounds on $\mu(\tilde{\bm{\xi}}, \sigma^!)$} 
Recall that $\mu(\tilde{\bm{\xi}}, \sigma^!)$ is the expected utility accrued from leaves $(\sigma_1, \sigma_2) \in \mathcal{L}$, where $\sigma_1 \succeq \sigma^!$. We can recursively decompose $\mu(\tilde{\bm{\xi}}, \sigma^!)$ into values of infosets, sequences and their summations. 
{
\small
\begin{align}
    d(\sigma; \sigma^!) &=\left(\mu(\blueprint, \sigma) - \breve{\mu}(\sigma; \sigma^!)\right)/\big| \{ I | \sigma(I) = \sigma \} \big| \label{eq:slack_mu}
    \\
    \breve{v}(I; \sigma^!) &= v(I) - d(\sigma(I); \sigma^!) \label{eq:lower_bound} \\ 
    f(I;\sigma^!) &= (v(\blueprint, I) - \breve{v}(I; \sigma^!))/\big | \mathcal{A}(I) \big| \label{eq:lb_slack_infoset}\\
    \breve{\mu} ( \sigma; \sigma^!) &= \mu(\blueprint, \sigma) - f(I;\sigma^!) \qquad I:\sigma=(I,a) \label{eq:lb_mu}
\end{align}
}
where $v(\blueprint, I) = \sum_{\sigma:(I,a), a\in\mathcal{A}(I)}{\mu(\blueprint}, \sigma)$, For every $\sigma$, starting from $\sigma^!$, we compute in \eqref{eq:slack_mu} the \textit{slack}, i.e., the difference between our desired lower bound $\breve{\mu}(\sigma)$ and what was achieved with the blueprint. In \eqref{eq:lower_bound}, this slack is split equally between all infosets which have $\sigma$ as the parent sequence. A similar process is repeated for infosets in \eqref{eq:lb_slack_infoset} and \eqref{eq:lb_mu}. We alternate between computing lower bounds for sequences and infosets until we have computed $\breve{v}(I)$ for $I \in \mathcal{I}^{\text{head}}_{i}$ in \eqref{eq:lower_bound}. We repeat this for all $\sigma^! \in \hat{\game}$ and take the tighter of the bounds to obtain $\breve{v}(I^{\text{head}})$ for all $I^{\text{head}} \in \mathcal{I}^{\text{head}}_{i}$.
\paragraph{(C) Upper bounds on $\beta^*(\tilde{\bm{\xi}}; \sigma^!)$} Recall that $\beta^*$ is the value of the best-deviating response. We can unroll the inequalities using Definition~\ref{def:exploitability} and stop once a head infoset is reached, i.e., when we encounter a term $\nu(I, \tilde{\bm{\xi}}; \sigma^!)$ for some $I\in \mathcal{I}^{\text{head}}_{i}$. If these terms were upper-bounded appropriately, then $\beta^*(\tilde{\bm{\xi}}; \sigma^!)$ would be upper-bounded.
One possible way is to compute upper bounds recursively
{
\small 
\begin{align}
s(\sigma; \sigma^!) &= \left( \invbreve{\beta}(\sigma; \sigma^!) - \beta(\sigma,\blueprint; \sigma^!)\right) /  \big| \{ I | \sigma(I) = \sigma \} \big|\\
\invbreve{\nu}(I; \sigma^!) &= 
\nu(I, \blueprint; \sigma^!) + s  
\label{eq:upper_bound}
\\
\invbreve{\beta}(\sigma; \sigma^!) &= \invbreve{\nu}(I; \sigma^!).
\end{align}
}
At the end of the bounds computation step, we have sets
$\invbreve{\mathcal{B}}_{i, j} = \{ 
(I, \sigma^!, \invbreve{\nu}(I; \sigma^!)
\}$
and 
$\breve{\mathcal{B}}_{i, j} = 
\{ 
(I, \breve{v}(I))
\}
$,
containing constraints of the form $v(\tilde{\bm{\xi}}; I^{\text{head}}) \geq \breve{v}(I^{\text{head}})$ and 
$\nu(I^{\text{head}}, \tilde{\bm{\xi}}; \sigma^!) \leq \invbreve{\nu}(I; \sigma^!)$
for some $I \in \mathcal{I}^{\text{head}}_{i, j}$.
\begin{theorem}
If $\tilde{\bm{\xi}}$ satisfies all constraints in $\breve{\mathcal{B}}_{i, j}$ and $\invbreve{\mathcal{B}}_{i, j}$ for all $i \in \{1, 2 \}$ and $j \in [J]$, then     $\delta^* (\tilde{\bm{\xi}}; \sigma^!)  \leq \max \left(0,  \delta^* 
    (\blueprint; \sigma^!) \right)$ for all $\sigma^! \in \hat{\game}$.
    \label{thm:satisfies_bounds}
\end{theorem}

\paragraph{Piecing the LP together}
Once these bounds are computed, enforcing them is simply a matter of placing them on top of the constraints for trigger sequences in $\check{\game}_j$ in \eqref{eq:incentive_inside_subgame}. For lower bounds of the form $(I, \breve(I)) \in \breve{\mathcal{B}}_{i, j}$, we introduce variables $v(I)$, where $v(I) = \sum_{\sigma:(I,a), a\in\mathcal{A}(I)}\mu(\sigma)$ and enforce $v(I) \geq \breve{v}(I)$. Note that the auxiliary variables $\mu(\sigma)$ has already been introduced as part of exploitability of $\check{\game}_j$ when enforcing \eqref{eq:incentive_inside_subgame}. For upper bounds $(I, \sigma^!, \invbreve{\nu}(I; \sigma^!)) \in \invbreve{\mathcal{B}}_{i, j}$, we introduce variables $\nu(I; \sigma^!)$ and enforce $\nu(I; \sigma^!) \leq \invbreve{\nu}(I; \sigma^!)$. To ensure that $\nu(I; \sigma^!)$ is indeed the value of the infoset given a trigger sequence $\sigma^!$, we will have to introduce auxiliary variables similar to \eqref{eq:lp_beta}, \eqref{eq:lp_nu} recursively. A summary and more precise explanation is included in the appendix. This LP is always feasible, since the blueprint would trivially satisfy all bounds constraints. To achieve full safety, we simply set the objective to be the component of social welfare culminating from subgame $j$. 
\section{Subgame Resolving with Regret Minimization}
Our second algorithm is based on regret minimization. 
We solve a saddle-point problem using self-play, utilizing the scaled extension operator of \citet{farina2019efficient} to provide an efficient regret minimizer over $\polytopesub$. This leads to a significantly more efficient algorithm.

\paragraph{Refinements as a Bilinear Saddle-point Problem} 
First, we show that the refinement LP may be written as a bilinear saddle point problem, similar to 
what was done in \citet{farina2019correlation}.
Observe that a refinement $\plansub$ is safe if and only if the greatest violation of the safety constraints to be equal to $0$. Building on this intuition, we introduce for each safety constraint, multipliers $\lambda^\delta_{i,\sigma^!}$, $\lambda^\nu_{i,I,\sigma^!}$ and $\lambda^v_{i,I}$ --- for exploitability (in $\check{\game}_j$), upper bounds, and lower bounds respectively. These multipliers are non-negative and sum to 1. Additionally, we introduce variables $\check{y}_{i,\sigma^!} \in \check{Y}_{i, \sigma^!}$ for $(I^!, a^!) = \sigma^! \in \check{\Sigma}_{i,j}$. Similarly, for trigger sequences $(I^!, a^!) = \sigma^! \in \sequencepre$, we introduce $\hat{y}_{i,\sigma^!} \in \hat{Y}_{i, \sigma^!}$. These $y$'s represent the components of the best-deviating responses to trigger sequences $\sigma^!$, and whose polytopes can be easily represented using the sequence-form representation of \citet{von1996efficient}. 
We explain in more detail in the Appendix. Resolving is equivalent to solving the following bilinear saddle point problem: 
{
\small 
\begin{align}
    \min_{\plansub}
    \max_{
    i, \lambda,y
    }
    \left \{ 
    \begin{array}{c}
        \sum\limits_{i, \sigma^! \in \check{\Sigma}_{i, j}} 
        \left[ \plansub^T R^\delta z^\delta_{i, \sigma^!} + \plansub^T \left( \lambda^{\delta}_{i, \sigma^!} b^\delta_{i, \sigma^!} \right) \right] + \\
        \sum\limits_{\substack{i, (I, \sigma^!, \cdot) \\ \in \invbreve{\mathcal{B}}_{i, j}}} 
        \left[ \plansub^T R^\nu z^\nu_{i, \sigma^!} + \plansub^T \left( \lambda^{\nu}_{i, \sigma^!} b^{\nu}_{i, \sigma^!} \right) \right] + \\
        \sum\limits_{i, (I, \cdot) \in \breve{B}_j} \plansub^T \left( \lambda^v_{i, I} b^v_{i, I} \right)
    \end{array}
    \right \}, \label{eq:minmax_bilinear}
\end{align}
}
where $z^\delta_{i, \sigma^!} = \lambda^{\delta}_{i, \sigma^!} \check{y}_{i, \sigma^!}$ and $z^\delta_{i, \sigma^!} = \lambda^{\nu}_{i, \sigma^!} \hat{y}_{i, \sigma^!}$, for appropriately chosen constants $R, b$ (which may vary on $\blueprint$). Hence, we can treat the refinement problem as a \textit{zero-sum} game between a \textit{mediator}, who chooses a refinement $\plansub$ and \textit{deviator}, who chooses multipliers and best-deviating responses. This zero-sum game can be solved by running self-play between two Hannan-consistent regret minimizers and taking average strategies. A regret minimizer for the deviator can be constructed efficiently using counterfactual regret minimization \citep{zinkevich2007regret}. A regret minimizer over $\polytopesub$ is constructed using the decomposition technique used by \citet{farina2019efficient} with some additional tiebreaking rules to ensure we do not have to "fill-in" sequence pairs in $\hat{\game}$. 
The algorithm is outlined in Algorithm~\ref{alg:regret_sr}, with the full details of the modified decomposition algorithm and deviator polytope presented in the appendix.


\begin{algorithm}[tb]
\caption{Refinement with Regret Minimization}
\label{alg:regret_sr}
\textbf{Input}: EFG, blueprint $\blueprint$
\begin{algorithmic}[1] 
\STATE Decompose $\polytopesub$ into series of scaled extensions.
\STATE Construct regret RM'er $\mathcal{X}$ over $\polytopesub$.
\STATE Construct regret RM'er $\mathcal{Y}$ over deviators.
\WHILE{saddle point gap $\geq \epsilon$}
    \STATE $\plansub^{(t)} \leftarrow \mathcal{X}$.recommend; $y^{(t)} \leftarrow\mathcal{Y}$.recommend
    \STATE $\mathcal{X}$.observeLoss($y_t$); $\mathcal{Y}$.observeLoss($\xi_t$)
\ENDWHILE
\end{algorithmic}
\end{algorithm}



\section{Experiments}
We evaluate our algorithms using the LP-based and regret minimization-based refining. We use the benchmark game of EFCE called \textit{Battleship}, introduced by \citet{farina2019correlation}. This game is played in 2 stages. In the \textit{placement} stage, players privately place their ship(s) of size 1 by $m$ on $W \times H$ grid. In the \textit{firing} stage, players take turns firing at each other over $T$ timesteps, or until a player's ship is destroyed. Each shot is at a single tile, and a ship is considered destroyed when all tiles in the ship are shot at least once. 
A player gets $1$ point for destroying the opponent's ship, but loses $\gamma$ points if his ship is destroyed. If no ship is destroyed by the end of the game, the game ends in a tie and both players get $0$. 

We use 2 different correlation blueprints for our experiments, \textit{Uniform} and \textit{Jittered}. Both correlation plans are based on independent player strategies stored using the sequence form. That is, $\blueprintstrat[\sigma_1, \sigma_2] = \blueprintstrat^{(1)}(\sigma_1) \cdot \blueprintstrat^{(2)}(\sigma_2)$, where $\blueprintstrat^{(1)}, \blueprintstrat^{(2)}$ are sequence form strategies for each player. In \textit{Uniform}, $\blueprintstrat^{(i)}$ have actions uniformly at random at each infoset. In \textit{Jittered}, each player has randomly generated behavioral strategies. Here, for infoset $I \in \mathcal{I}_i$, action $a_j \in \mathcal{A}(I)$ is played with probability $p(a_j; I)= \kappa_{I,j} / \sum_k \kappa_{I, k}$, with $\kappa_{I,k} = 1 + w \cdot \varepsilon_{I,k}$, where each $\varepsilon_{I,k}$ is drawn independently and uniformly from $[-1, 1]$ and $w \in [0, 1]$ is a width parameter governing the level of deviation from uniform strategies. 

Subgames are defined based on public information, which at the $k$-th step of firing are precisely the locations fired by each player. We base subgames on the shot history up till timestep $T'<T$. $T'$ balances the trade-off between accuracy versus computational costs. For a grid of size $n$, we have  
$J=\prod_{k = T-T'+1}^{T} k^2$ subgames. When $T'$ is small, we have fewer subgames, but can achieve better social welfare. All experiments are run on an Apple M1 Chip with 16GB of RAM with 8 cores. LPs are solved using Gurobi \cite{gurobi}.

\paragraph{Safe resolving with SW maximization}
We first show using our LP-based method that ensures fully safe resolving can lead to significantly higher social welfare as compared to the blueprint. We set $T'=1$ and we use ships with $m=1$, i.e., the game is over once any ship is hit. Consequently, the game is entirely symmetric in terms of location. The NE here is to play and shoot uniformly at random. Hence, \textit{Uniform} is a valid, though not SW-optimal EFCE. Under \textit{Uniform}, the exploitability $\delta^*$ under the blueprint is $0$, implying that the complete refinement $\tilde{\bm{\xi}}$ is also an EFCE. We perform refinement on the first subgame (this without loss of generality due to symmetry) and compare the SW accumulated from the subgame under the blueprint and refinement. For \textit{Jittered}, we repeated the experiment $10$ times with different seeds and report the mean. The results are reported in Table~\ref{tbl:bs_lp}. 
In all our experiments, our refined strategy $\plansub$ gives a much higher SW. For example, in the largest example with $\gamma=2$, SW increases by 4.6e-3. \textit{This is not a negligible improvement}; since this is applied to all $36$ subgames, the expected improvement in SW of the complete refinement $\tilde{\bm{\xi}}$ is actually 0.167. $|\polytopesub|$ is significantly smaller than $|\Xi|$, such that each refinement is computed in no more than 10 seconds.

\begingroup
\setlength{\tabcolsep}{5pt}
\begin{table}[t]
\centering
\begin{tabular}{ccccccc}
 $n$, $T$  & \multirow{2}{*}{$|\polytopesub|$} & \multirow{2}{*}{$\gamma$} &   \multicolumn{2}{c}{\textit{Uniform}} & \multicolumn{2}{c}{\textit{Jittered}} \\
  $J$ & & & BP & Refined & BP & Refined \\
 \hline \hline 
 $3, 2$, & \multirow{2}{*}{382} & 2&-3.70 & -3.70 &-3.55  & -3.55 \\
  9&   & 5& -14.8 & -14.8 & -14.2 & -14.2 \\
  \hline
  $4, 3$, &\multirow{2}{*}{3.2e3} & 2& -3.13 & -2.95 & -3.24 & -3.10 \\
  16 &   & 5 & -12.5 & -11.4 & -13.0 & -11.8 \\
  \hline
  $5, 3$, &\multirow{2}{*}{2.3e4} &2  & -1.92 & -1.34 & -1.95 & -1.25 \\
   25 &   & 5 & -7.68 & -4.80 & -7.82 & -4.32 \\
  \hline
  $6, 3$, & \multirow{2}{*}{1.2e5}& 2& -1.23 & -.772 & -1.25 & -.627 \\
  36&   & 5& -4.94 & -2.47 & -4.99 & -1.95   \\
  \hline
\end{tabular}
\caption{Comparison of social welfare between blueprint (BP) and SW-maximizing safe refinement with ships of size $1$. Social welfare is reported at a scale of 1e-2.}
\label{tbl:bs_lp}
\end{table}
\endgroup

\paragraph{Safe resolving using regret minimization}
We now demonstrate the scalability of refinement based on regret minimization. Our goal here is to demonstrate that subgame resolving can be performed efficiently for games that are too large for $\bm{\xi}$ to even be stored in memory. We run refinement using our regret minimization algorithm and report the "pseudo"-exploitability of $\plansub$ (i.e., the value of the inner maximization over $(i, \lambda, y)$,  \eqref{eq:minmax_bilinear}, or the most violated incentive constraint of the LP). We use $T'=1, \gamma=2$ and the \textit{Uniform} blueprint. The results are reported in Figure~\ref{fig:regret}. Our huge instance is several times larger than the largest instance in \citet{farina2019efficient}, and it would require a significant amount of memory to store a full correlation plan $\tilde{\bm{\xi}}$ . We find that in practice, resolving requires less than $0.5$ seconds per iteration, while using no more than 2GB of memory.

\begin{figure}
    \begin{minipage}{0.45\linewidth}
    \includegraphics[scale=0.27]{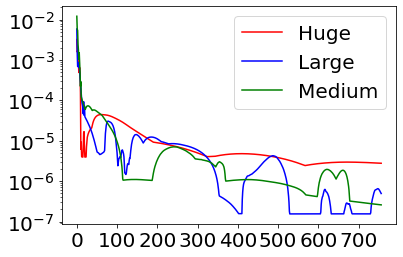}
    \end{minipage}
    \begin{minipage}{0.4\linewidth}
    \begingroup
    \setlength{\tabcolsep}{3pt}
    \begin{tabular}{cccc}
        & Med. & Large & Huge \\
        \hline \hline
        $W,H$  & 3,2 &3,2 &3,2 \\
        \hline
        $T$ &4 &4 &5\\
        \hline
        $m$ & 1&2 &2 \\
        \hline
        $|\Xi|$ & 3.89M &111M & 360M \\
        \hline
    \end{tabular}
    \endgroup
    \end{minipage}
    \caption{Left: Most violated incentive constraint of $\tilde{\bm{\xi}}$ plot against iteration number. Right: Parameters of game.}
    \label{fig:regret}
\end{figure}
\section{Conclusion}
In this paper, we propose a novel subgame resolving technique for EFCE. We offer two algorithms, the first based on LPs and the second uses regret minimization, both of which consume significantly less compute than full-game solvers. Our technique is, to the best of our knowledge, the first \textit{online} algorithm towards solving EFCE. In future, we hope to expand our work to other equilibria, such as those involving hindsight rationality \cite{morrill2021efficient}.

\section{Acknowledgements}
This work was supported in part by NSF grant IIS-2046640 (CAREER) and a research grant from Lockheed Martin.


\bibliography{aaai22.bib}

\begin{thebibliography}{16}
\providecommand{\natexlab}[1]{#1}

\bibitem[{Brown and Sandholm(2017)}]{brown2017safe}
Brown, N.; and Sandholm, T. 2017.
\newblock Safe and nested endgame solving for imperfect-information games.
\newblock In \emph{Workshops at the thirty-first AAAI conference on artificial
  intelligence}.

\bibitem[{Brown and Sandholm(2018)}]{brown2018superhuman}
Brown, N.; and Sandholm, T. 2018.
\newblock Superhuman AI for heads-up no-limit poker: Libratus beats top
  professionals.
\newblock \emph{Science}, 359(6374): 418--424.

\bibitem[{Dud{\'i}k and Gordon(2009)}]{dudik2012sampling}
Dud{\'i}k, M.; and Gordon, G.~J. 2009.
\newblock A Sampling-Based Approach to Computing Equilibria in Succinct
  Extensive-Form Games.
\newblock In \emph{UAI}.

\bibitem[{Farina, Kroer, and Sandholm(2019)}]{farina2019regret}
Farina, G.; Kroer, C.; and Sandholm, T. 2019.
\newblock Regret circuits: Composability of regret minimizers.
\newblock In \emph{International conference on machine learning}, 1863--1872.
  PMLR.

\bibitem[{Farina et~al.(2019{\natexlab{a}})Farina, Ling, Fang, and
  Sandholm}]{farina2019correlation}
Farina, G.; Ling, C.~K.; Fang, F.; and Sandholm, T. 2019{\natexlab{a}}.
\newblock Correlation in Extensive-Form Games: Saddle-Point Formulation and
  Benchmarks.
\newblock In Wallach, H.; Larochelle, H.; Beygelzimer, A.; d\textquotesingle
  Alch\'{e}-Buc, F.; Fox, E.; and Garnett, R., eds., \emph{Advances in Neural
  Information Processing Systems}, volume~32. Curran Associates, Inc.

\bibitem[{Farina et~al.(2019{\natexlab{b}})Farina, Ling, Fang, and
  Sandholm}]{farina2019efficient}
Farina, G.; Ling, C.~K.; Fang, F.; and Sandholm, T. 2019{\natexlab{b}}.
\newblock Efficient Regret Minimization Algorithm for Extensive-Form Correlated
  Equilibrium.
\newblock In Wallach, H.; Larochelle, H.; Beygelzimer, A.; d\textquotesingle
  Alch\'{e}-Buc, F.; Fox, E.; and Garnett, R., eds., \emph{Advances in Neural
  Information Processing Systems}, volume~32. Curran Associates, Inc.

\bibitem[{Farina and Sandholm(2020)}]{farina2020polynomial}
Farina, G.; and Sandholm, T. 2020.
\newblock Polynomial-time computation of optimal correlated equilibria in
  two-player extensive-form games with public chance moves and beyond.
\newblock \emph{arXiv preprint arXiv:2009.04336}.

\bibitem[{Gray et~al.(2021)Gray, Lerer, Bakhtin, and Brown}]{gray2020human}
Gray, J.; Lerer, A.; Bakhtin, A.; and Brown, N. 2021.
\newblock Human-Level Performance in No-Press Diplomacy via Equilibrium Search.
\newblock In \emph{International Conference on Learning Representations}.

\bibitem[{{Gurobi Optimization, LLC}(2022)}]{gurobi}
{Gurobi Optimization, LLC}. 2022.
\newblock {Gurobi Optimizer Reference Manual}.
\newblock \url{https://www.gurobi.com}.
\newblock Accessed: 2022-04-01.

\bibitem[{Lerer et~al.(2020)Lerer, Hu, Foerster, and
  Brown}]{lerer2020improving}
Lerer, A.; Hu, H.; Foerster, J.; and Brown, N. 2020.
\newblock Improving policies via search in cooperative partially observable
  games.
\newblock In \emph{Proceedings of the AAAI Conference on Artificial
  Intelligence}, volume~34, 7187--7194.

\bibitem[{Ling and Brown(2021)}]{ling2021safe}
Ling, C.~K.; and Brown, N. 2021.
\newblock Safe Search for Stackelberg Equilibria in Extensive-Form Games.
\newblock In \emph{Proceedings of the AAAI Conference on Artificial
  Intelligence}, volume~35, 5541--5548.

\bibitem[{Morrill et~al.(2021)Morrill, D'Orazio, Lanctot, Wright, Bowling, and
  Greenwald}]{morrill2021efficient}
Morrill, D.; D'Orazio, R.; Lanctot, M.; Wright, J.~R.; Bowling, M.; and
  Greenwald, A. 2021.
\newblock Efficient Deviation Types and Learning for Hindsight Rationality in
  Extensive-Form Games.
\newblock \emph{arXiv preprint arXiv:2102.06973}.

\bibitem[{Von~Stengel(1996)}]{von1996efficient}
Von~Stengel, B. 1996.
\newblock Efficient computation of behavior strategies.
\newblock \emph{Games and Economic Behavior}, 14(2): 220--246.

\bibitem[{Von~Stengel and Forges(2008)}]{von2008extensive}
Von~Stengel, B.; and Forges, F. 2008.
\newblock Extensive-form correlated equilibrium: Definition and computational
  complexity.
\newblock \emph{Mathematics of Operations Research}, 33(4): 1002--1022.

\bibitem[{Zhang and Sandholm(2021)}]{zhang2021subgame}
Zhang, B.~H.; and Sandholm, T. 2021.
\newblock Subgame solving without common knowledge.
\newblock \emph{arXiv preprint arXiv:2106.06068}.

\bibitem[{Zinkevich et~al.(2007)Zinkevich, Johanson, Bowling, and
  Piccione}]{zinkevich2007regret}
Zinkevich, M.; Johanson, M.; Bowling, M.; and Piccione, C. 2007.
\newblock Regret minimization in games with incomplete information.
\newblock \emph{Advances in neural information processing systems}, 20:
  1729--1736.

\end{thebibliography}

\newpage

\appendix
\section{Definition of Perfect Recall}
Formally, we require that (i) if $h$ and $h'$ belonging to $P_i$ are in distinct information sets $I_i$ and $I_i'$, then their descendants $s \sqsupset h$ and $s' \sqsupset h'$ never belong to the same information set, (ii) if states $h, h'$, possibly equal, belong to the same information set $I_i$, which contains distinct actions $a, a'$, then the states $s \sqsupseteq ha$ and $s' \sqsupseteq h'a'$ belong to different information sets.
\section{Proof of Theorem~\ref{thm:independence_subgames}}
\begin{proof}
The cases are disjoint. Hence, it suffices to show that there does not exist $\sigma_1 \relevant \sigma_2$, where $\sigma_1 = (I_1, a_1) \in \check{\game}_j, \sigma_2 =(I_2, a_2) \in \check{\game}_k$ for $j\neq k$. By definition $\sigma_1 \connected \sigma_2$ implies that there exists a path from the root passing through vertices $v_1 \in I_1$ and $v_2 \in I_2$. WLOG suppose that $v_1$ precedes $v_2$ in this path. Then $v_2$ must lie in $\check{\game}_j$ by property (i) of the subgame. This is a contradiction.
\end{proof}
\section{Full Description of LP}
In this section, we describe the LP used to compute a refinement $\plansub$ in subgame $j$ in detail. There are 2 classes of constraints in the LP -- the structural constraints an the incentive constraints. The \textit{structural constraints} enforce that $\plansub$ lies in $\polytopesub$. The \textit{incentive constraints} consist of three sets of constraints (A-C) which ensure safety. (A) ensures that safety is achieved for in-subgame triggers, while (B) and (C) ensures safety for triggers that lie in $\hat{\game}$.
\subsection{Structural Constraints}
Let us first describe the indices of $\polytopesub$. These are sequence pairs $(\sigma_1, \sigma_2)$, where $\sigma_i$ is either $\varnothing$ or $(I_i, a_i)$. As with $\Xi$, we require $\sigma_1 \relevant \sigma_2$. But, we also require that none of $\sigma_i$ lies in another $\game_k$, $k\neq j$, i.e., $S_k$ in Theorem~\ref{thm:independence_subgames}. That is, if $\sigma_i$ is non-empty and $\sigma_i = (I_i,a_i)$, then $I_i$ lies in subgame $j$. Naturally, we require that $\plansub \geq 0$, and that $\plansubstrat[\varnothing, \varnothing] = 1$.

\paragraph{Sequence-form constraints on rows and columns of $\plansub$.}
Next, we will enforce the sequence-form constraints on the rows and columns of $\polytopesub$. These are similar to the flow conservation constraints, but are for the probability of sequence pairs. For example, in the left subgame game in Figure~\ref{fig:signalling}, $\polytopesub$ is determined by the first 3 columns, and we have row constraint $\xi[\sigma_1, \ell_x] + \xi[\sigma_1, r_x] = \xi[\sigma_1, \varnothing]$, and column constraints $\xi[G, \sigma_2] + \xi[B, \sigma_2] = \xi[\varnothing, \sigma_2]$,
$\xi[X_G, \sigma_2] + \xi[Y_G, \sigma_2] = \xi[G, \sigma_2]$, and 
$\xi[X_B, \sigma_2] + \xi[Y_B, \sigma_2] = \xi[B, \sigma_2]$ for $\sigma_2 \in \{ \varnothing, \ell_x, r_x\}$. Generally, we have
\begin{align*}
    \polytopesub := 
    \left\{
    \bm{\xi} \geq 0: 
    \begin{array}{ccc}
    \bm{\xi} [\varnothing, \varnothing] = 1, \\ 
    \sum_{a \in \mathcal{A}(I)} \xi[(I_1,a), \sigma_2] = \xi[\sigma(I_1), \sigma_2], \\
    \sum_{a \in \mathcal{A}(I)} \xi[\sigma_1, (I_2, a)] = \xi[\sigma_1, \sigma(I_2)]
    \end{array}
    \right\},
\end{align*}
The constraints defining $\Xi_j$ are essentially identical to that of $\Xi$ -- the convex polytope of correlation plans in the original game, except we now work with a restricted index set.

\paragraph{Equality-to-blueprint constraints.}
We now examine the equality-to-blueprint constraints. This ensures that $\plansub$ is indeed a refinement of $\blueprintstrat$. In Figure~\ref{fig:signalling}, this would mean that the entries in $\hat{\game}$ (entries in the first column) is equal to the blueprint, i.e., $\plansubstrat[\sigma_1, \varnothing] = \blueprintstrat[\sigma_1, \varnothing]$ for all $\sigma_1 \in \Sigma_1$. More generally, for all $(\sigma_1, \sigma_2) \in S_0$ (i.e., $\sigma_1, \sigma_2$ are both either equal to $\varnothing$ \textit{both} or do not belong to a subgame), we require that $\plansubstrat[\sigma_1, \sigma_2] = \blueprintstrat[\sigma_1, \sigma_2]$.

\subsection{Auxiliary variables}
Enforcing constraint set (A-C) requires the introduction of numerous variables, which can be a little daunting at first glance. We first lay them out here for clarity. For this section, unless otherwise stated, these variables are with respect to $\plansub$. We will make it explicit if we need to reference the blueprint $\blueprint$.

\paragraph{Values of sequences assuming no deviation.} The first are the variables $\mu(\sigma)$, which exist for each player for all $\sigma = (I, a) \in \Sigma_i$, where $I \in \check{I}_j$. These capture the value of contribution of payoffs for $P_i$ \textit{assuming both players abide to all recommendations} under $\plansubstrat$ for all leaves involving sequences $\sigma' \succeq \sigma$. These can be computed recursively the same way as the original LP of \citet{von2008extensive} in \eqref{eq:lp_mu}. 
\begin{align*}
    \mu(\sigma) &= \sum_{\sigma_2; (\sigma, \sigma_2) \in \mathcal{L}} u_1(\sigma, \sigma_2) \plansubstrat[\sigma, \sigma_2] + \sum_{\sigma' \succ_1 \sigma} \mu(\sigma') ,
\end{align*}
that is, value of each sequence $\sigma_i$ is given by the rewards for $P_i$ from leaves containing $\sigma_i$, plus the rewards from future sequences (computed recursively). Note that by the definition of a subgame (Definition~\ref{def:subgame}), during the recursive process, we will never have to `address' a $\mu(\sigma')$ where $\sigma'$ lies outside of $\check{\Sigma}_{j}$. The number of variables and constraints is approximately $|\check{\Sigma}_{j}|$. We will eventually use these in constraints for in-subgame triggers (A), as well as the lower bounds in (B).

\paragraph{Value of infosets assuming no deviations.} For convenience, we will write the value of infosets of $I$ as $v(I) = \sum_{\sigma: (I, a), a \in \mathcal{A}(I)} \mu(\sigma)$. These are used for the lower bounds in (A). In our implementation, they are also used as auxiliary variables while recursively computing $\mu(\sigma)$. This is equivalent to the definition of $\mu(\sigma)$ provided above. The number of variables here is $|\check{\mathcal{I}}_j|$ (which is in turn upper bounded by $|\check{\Sigma}_{j}|$.

\paragraph{Values of infosets under deviations.} Next, we have the variables $\nu(I; \sigma^!)$. These represent the values of infoset $I$ given that the player was triggered by $\sigma^!$, deviated and plays the best response after his deviation. Let $\sigma^! = (I^!, a^!)$ be a trigger sequence, and $\sigma' = (I^!, a'), a' \neq a^!$ is the sequence which was deivated to. $v(I; \sigma^!)$ exists for each $I \in \mathcal{I}_j$ where $\sigma(I) \succeq \sigma'$, i.e., $I$ (which belongs in the subgame $j$) could be encountered after deviating to $\sigma'$, which lies in the same infoset as $\sigma^!$. The variables in $\nu$ can be further split into 2 groups. If $\sigma^!$ lies in $\check{\Sigma}_{j}$, then this is similar to the variable in \eqref{eq:lp_nu}. If not, then note that these are only created for infosets within subgame $j$. The former is used in enforcing safety for in-subgame triggers (A) and the latter for constraint set (C). The total number of variables here for each $P_i$ is no greater (and in practice, usually much smaller) than $(|\hat{\Sigma}_i| + |\check{\Sigma}_{i,j}| ) \cdot |\check{\mathcal{I}}_{i,j}|$.

\paragraph{Value of sequences under deviations.} $\beta(\sigma; \sigma^!)$ is very similar to $\nu$, in that it is the value of a sequence assuming the last recommendation received and deviated from was $\sigma^!$. Again, we are only concerned with sequences $\sigma \in \check{\Sigma}_{j}$, and those which could be reached after deviation from $\sigma^!$. 
For a fixed trigger sequence $\sigma^!$, $\beta(\sigma; \sigma^!)$ and $\nu(I; \sigma^!)$ can be computed together recursively using Equations~\eqref{eq:lp_nu} and \eqref{eq:lp_beta}. The inequalities ensure that the values in $\nu$ and $\beta$ are such that these are no less than best-responses towards $\plansub$. The number of variables here is no greater than $(|\hat{\Sigma}_i| + |\check{\Sigma}_{i,j}| ) \cdot |\check{\Sigma}_{i,j}|$.

\subsection{Incentive Constraint Set (A)}
Recall that these ensure that safety is achieved for $\sigma^! \in \check{\Sigma}_j$. For each such $\sigma^!$, we have constraints of the form (as in the main paper)
\begin{align}
    \mu(\tilde{\bm{\xi}}; \sigma^!) \geq \beta^*(\tilde{\bm{\xi}}; \sigma^!) - \delta^*(\blueprint; \sigma^!).
\end{align}
\subsection{Incentive Constraint Set (B)}
We showed in the main paper that these are intended to guarantee lower bounds on $\mu(\sigma^!)$, where $\sigma \in \hat{\Sigma}_i$. The sufficient conditions for doing so are in the set $\breve{\mathcal{B}}_{i,j}$, each of the form $(I, \breve{v}(I))$. As mentioned in the main paper, we will need 
\begin{align*}
    v(I^{\text{head}}) \geq \breve{v}(I^{\text{head}})
\end{align*}
for some $I^{\text{head}} \in \mathcal{I}^{\text{head}}_{i,j}$.
\subsection{Incentive Constraint Set (C)}
Similarly, (C) is intended to upper bound $\beta^*(\sigma^!)$, i.e., ensure that deviating would not be too beneficial. This is achieved by making sure $\beta^*(\sigma; \sigma^!)$ for all $\sigma \neq \sigma^!$ and $\sigma=(I^!,a)$. As we showed in the main paper, this can be achieved when $\nu(I^\text{head}; \sigma^!) \leq \invbreve{\nu}(I, \sigma^!)$ for each tuple $(I, \sigma^!, \invbreve{v}(I; \sigma^!)) \in \invbreve{\mathcal{B}}_{i,j}$. 

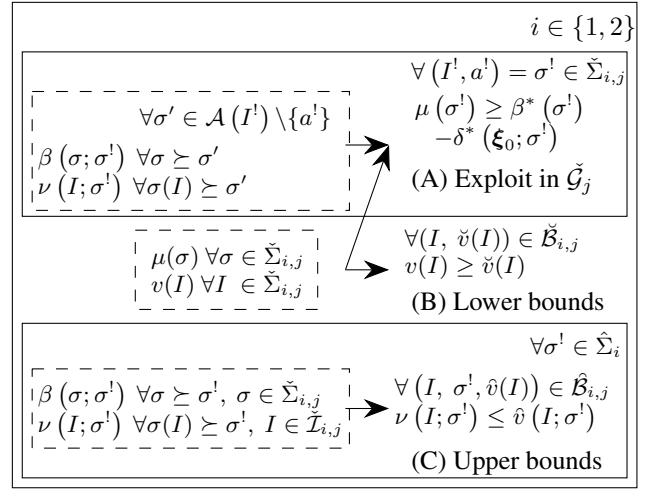
\begin{figure}
\begin{tikzpicture}[x=0.75pt,y=0.75pt,yscale=-1,xscale=1]
\draw   (0,-5) -- (315,-5) -- (315,240) -- (0,240) -- cycle ;
\draw   (5,157) -- (310,157) -- (310,235) -- (5,235) -- cycle ;
\draw   (5,20) -- (310,20) -- (310,103) -- (5,103) -- cycle ;
\draw   [-{Stealth[scale=2]}] (168,67) -- (190,67) ;
\draw   [-{Stealth[scale=2]}] (168,130) -- (190,130) ;
\draw   [-{Stealth[scale=2]}] (168,200) -- (190, 200) ;
\draw   [-{Stealth[scale=2]}] (168,130) -- (190, 67) ;
\draw  [dash pattern={on 4.5pt off 4.5pt}] (10,40) -- (170,40) -- (170,100) -- (10,100) -- cycle ;
\draw  [dash pattern={on 4.5pt off 4.5pt}] (10,180) -- (170,180) -- (170,220) -- (10,220) -- cycle ;
\draw  [dash pattern={on 4.5pt off 4.5pt}] (62,110) -- (154,110) -- (154,150) -- (62,150) -- cycle ;

\draw (195,40) node [anchor=north west][inner sep=0.75pt]  [font=\small]  {$
\begin{array}{c}\mu \left(\sigma ^{!}\right) \geq \beta ^{*}\left(\sigma ^{!}\right)\\ -\delta ^{*}\left(\blueprint ;\sigma ^{!}\right) \ 
\end{array}
$};
\draw (260,0) node [anchor=north west][inner sep=0.75pt]    {$i \in \{1,2\} \ $};
\draw (62,115) node [anchor=north west][inner sep=0.75pt]  [font=\small]  {$ \begin{array}{l}
\mu ( \sigma ) \ \forall \sigma \in \check{\Sigma }_{i,j}\\
v( I) \ \forall I\ \in \check{\Sigma}_{i,j}
\end{array}$};
\draw (5,65) node [anchor=north west][inner sep=0.75pt]  [font=\small]  {$ \begin{array}{l}
\beta \left( \sigma ;\sigma ^{!}\right) \ \forall \sigma \succeq \sigma '\ \\
\nu \left( I ;\sigma ^{!}\right) \ \forall \sigma ( I) \succeq \sigma '
\end{array}$};
\draw (260,160) node [anchor=north west][inner sep=0.75pt]  [font=\small]  {$\forall \sigma ^{!} \in \hat{\Sigma }_{i}$};
\draw (200,22) node [anchor=north west][inner sep=0.75pt]  [font=\small]  {$\forall \left( I^{!} ,a^{!}\right) =\sigma ^{!} \in \check{\Sigma }_{i,j}$};
\draw (5,185) node [anchor=north west][inner sep=0.75pt]  [font=\small]  {$ \begin{array}{l}
\beta \left( \sigma ;\sigma ^{!}\right) \ \forall \sigma \succeq \sigma ^{!} ,\ \sigma \in \check{\Sigma}_{i,j}\\
\nu \left( I;\sigma ^{!}\right) \ \forall \sigma ( I) \succeq \sigma ^{!} ,\ I\in \check{\mathcal{I}}_{i,j}
\end{array}$};
\draw (190,105.4) node [anchor=north west][inner sep=0.75pt]  [font=\footnotesize]  {$ \begin{array}{l}
\forall ( I,\ \breve{v}( I)) \in \breve{\mathcal{B}}_{i,j}\\
v( I) \geq \breve{v}( I)
\end{array}$};
\draw (185,180.4) node [anchor=north west][inner sep=0.75pt]  [font=\footnotesize]  {$ \begin{array}{l}
\forall \left( I,\ \sigma ^{!} ,\invbreve{v}( I)\right) \in \invbreve{\mathcal{B}}_{i,j}\\
\nu \left( I ;\sigma ^{!}\right) \leq \invbreve{v}\left( I;\sigma ^{!}\right)
\end{array}$};
\draw (63.11,45) node [anchor=north west][inner sep=0.75pt]  [font=\small]  {$\forall \sigma '\in \mathcal{A}\left( I^{!}\right) \backslash \{a^!\} \ $};
\draw (200,75) node [anchor=north west][inner sep=0.75pt]   [align=left] {(A) Exploit in $\check{\game_{j}}$};
\draw (200,140) node [anchor=north west][inner sep=0.75pt]   [align=left] {(B) Lower bounds};
\draw (200,220) node [anchor=north west][inner sep=0.75pt]   [align=left] {(C) Upper bounds};
\end{tikzpicture}
\caption{Summary of the incentive constraints required by a refinement $\plansub \in \polytopesub$, where the term $\plansub$ is omitted for brevity. Auxilary variables are on the left, while constraints are on the right. Dashed boxes are either computed recursively, either as a best-deviating responses to trigger sequences $\sigma^!$ or as the value of a sequence/infoset. Quantifiers in solid boxes are applied to all variables and constraints in the box. Arrows signify dependencies between variables and constraints.}
\label{fig:lp_diagram}
\end{figure}
Figure~\ref{fig:lp_diagram} summarizes the variables and constraints (excluding structural constraints).
\section{Proof of Theorem~\ref{thm:satisfies_bounds}.}
In the main paper, we showed that safety will be achieved for trigger sequence $\sigma^!$ with an appropriate choice of $\invbreve{\beta}(\sigma^!; \sigma^!)$ and $\breve{\mu}(\sigma^!; \sigma^!)$.

\subsection{Showing lower bounds $\mu(\blueprint, \sigma^!) \geq \breve\mu(\sigma^!; \sigma^!)$ hold.}
 Suppose for all $\sigma' \succ_1 \sigma$, where $\sigma \succeq \sigma^!$ we have $\mu(\planref, \sigma') \geq \breve{\mu}(\sigma'; \sigma^!)$. We show this implies that $\mu(\planref, \sigma) \geq  \breve{\mu}(\sigma; \sigma^!)$. 

\begin{align*}
    \mu(\planref, \sigma') - \breve{\mu} (\sigma; \sigma^!) \geq 0 \\
    \mu(\planref, \sigma') - \mu(\blueprint, \sigma) + f(I;\sigma^!) \geq 0 \\
    \mu(\planref, \sigma') - \mu(\blueprint, \sigma)+ 
    (v(\blueprint,I) - \breve{v}(I; \sigma^!))/\mathcal{A}(I) \geq 0
\end{align*}
Now we take summations over all $\sigma^!$ belonging to infoset $I$, where $\sigma(I) = \sigma$. This gives
\begin{align*}
    \sum_{\sigma'} (\mu(\planref, \sigma') - \mu(\blueprint, \sigma)+ 
    (v(\blueprint,I) - \breve{v}(I; \sigma^!))/(\mathcal{A}(I))) \geq 0 \\
    \sum_{\sigma'} (\mu(\planref, \sigma') - \mu(\blueprint, \sigma)+ 
    (v(\blueprint,I) - \breve{v}(I; \sigma^!))) \geq 0 \\
    \sum_{\sigma'} (\mu(\planref, \sigma') - \mu(\blueprint, \sigma))+ 
    v(\blueprint,I) - \breve{v}(I; \sigma^!) \geq 0 \\
    v(\planref, I)  - \breve{v}(I; \sigma^!) \geq 0,
\end{align*}
that is, the `subbound' of $\breve{v}$ is satisfied. We repeat a similar process for all $I$ such that $\sigma(I) = \sigma$.
\begin{align*}
    v(\planref, I)  - v(\blueprint, I) - d(\sigma;\sigma^!) \geq 0 \\
    v(\planref, I)  - v(\blueprint, I) - 
    \frac{
    \mu(\blueprint, \sigma) - \breve{\mu}(\sigma; \sigma^!)
    }{| \{ I |\sigma(I)=\sigma \} |} \geq 0.
\end{align*}
Taking summations over all such $I$,
\begin{align*}
    \sum_{I:\sigma(I) = \sigma} \left( v(\planref, I)  - v(\blueprint, I) \right) - 
    \mu(\blueprint, \sigma) - \breve{\mu}(\sigma; \sigma^!)
     \geq 0\\ 
     \mu(\planref, \sigma) - \breve{\mu}(\sigma; \sigma^!) \geq 0
\end{align*}
Applying this repeatedly starting from the bounds in $\breve{\mathcal{B}}$ gives us the required result. 

\subsection{Showing upper bounds $\beta^*(\tilde{\bm{\xi}}; \sigma^!) \leq \invbreve{\beta}(\sigma^!; \sigma^!)$ hold.}

As before, suppose for all $\sigma' \succ_1 \sigma$, where $\sigma \succeq \sigma^!$ we have $\beta(\sigma', \planref; \sigma^!) \leq \invbreve{\beta}(\sigma', \sigma^!)$. 
\begin{align*}
    \beta(\sigma', \planref; \sigma^!) \leq \invbreve{\beta}(\sigma', \sigma^!) \\
    \beta(\sigma', \planref; \sigma^!) \leq \invbreve{v}(I; \sigma^!) \\
    \nu(I, \planref; \sigma^!) \leq \invbreve{v}(I; \sigma^!) \\
    \nu(I, \planref; \sigma^!) \leq \nu(I, \blueprint; \sigma^!) + s \\
    \nu(I, \planref; \sigma^!) \leq \nu(I, \blueprint; \sigma^!) + 
    \frac{ \invbreve{\beta} (\sigma; \sigma^!) - \beta(\sigma, \blueprint; \sigma^!)}{| \{ I|\sigma(I) = \sigma \} |}
\end{align*}
Summing over all $I$ where $\sigma(I)=\sigma$ and the definiion of $\beta$ in Definition~\ref{def:exploitability} gives
\begin{align*}
    \beta(\sigma, \planref; \sigma^!) \leq  
    \invbreve{\beta} (\sigma; \sigma^!).
\end{align*}
As before, we begin with $\nu(I, \planref; \sigma^!)$ meeting the bounds in $\invbreve{\mathcal{B}}$. We then repeatedly apply these derivations repeatedly until we have $\beta(\sigma, \planref; \sigma^!) \leq \invbreve{\beta} (\sigma; \sigma^!)$ for all $\sigma = (I^!, a), a \neq a^!$. This implies $\beta^*(\planref; \sigma^!) \leq \invbreve{\beta} (\sigma; \sigma^!)$ as required.

Combining these results together with the discussion in the main paper completes the result.
\section{Refinements as a Bilinear Saddle-point Problem}
The LP \textit{without the objective} (and by extension, safe resolving) can be written as a bilinear saddle point problem. Consider a refinement $\plansub$. The largest violation of a constraint is:
\begin{align*}
    \max_{i \in \{ 1, 2 \}}
    \left\{
    \begin{array}{c}
    \max\limits_{\sigma^! \in \check{\Sigma}_{i, j}}
    \delta^*(\tilde{\bm{\xi}}; \sigma^!) - \delta^*(\blueprint; \sigma^!) \\
    \max\limits_{(I, \sigma^!, \invbreve{\nu}(I; \sigma^!) )\in \invbreve{\mathcal{B}}_{i,j}} \nu(I, \plansub; \sigma^!) - \invbreve{\nu}(I; \sigma^!) \\
    \max\limits_{(I, \breve{v}(I)) \in \breve{\mathcal{B}}_{i,j} } \breve{v}(I) - v(\plansub; I) \\
    \end{array}
    \right\}.
\end{align*}
The inner maximizations are for (A) safety for trigger sequences $\sigma^! \in \check{\game}_j$ (C) upper bounds on values head infoset $I$ under the best deviating response to a the trigger sequences $\sigma^!$, and (B) lower bounds on values of head infosets assuming no deviation occurs. If the $\plansub$ satisfies the LP, then the above expression is non-positive. These nested maximizations can be rewritten as the maximization of a linear function over a polytope with a polynomial number of constraints. For each $i \in \{ 1,2 \}$ we introduce multipliers for each of the maximizations: $\lambda^\delta_{i, \sigma^!} \forall \sigma^! \in \check{\Sigma}_{i,j}, \lambda^{\nu}_{i, I, \sigma^!}, \lambda^v_{i, I}$. 
{
\small 
\begin{align*}
    \max_{
    \substack{
    i, \lambda^{\delta}_{i, \sigma^!}, \lambda^{\nu}_{i, I, \sigma^!}, \lambda^v_{i, I}\\
    \sum(\sum \lambda^{\delta} + 
         \sum \lambda^{\nu} +
         \lambda^v = 1
    }
    }
    \left \{
    \begin{array}{c}
        \lambda^{\delta}_{i, \sigma^!} (\delta^*(\plansub; \sigma^!) - \delta^*(\blueprint; \sigma^!)) + \\
        \lambda^{\nu}_{i, I, \sigma^!} (\nu(I, \plansub; \sigma^!) - \invbreve{\nu}(I; \sigma^!)) + \\
        \lambda^v_{i, I} (\breve{v}(I) - v(\plansub; I))
    \end{array}
    \right \}
\end{align*}
}
Optimizing over $\nu$ is simply finding the best-deviating response to a trigger sequence $\sigma^!$ over a polytope $Y_{i, \sigma^!}$ using the sequence form representation \citep{von1996efficient}. 
Given $\mu$ as well as the bounds are constants, the expression can be written as a single \textit{linear} maximization over the multipliers and $y \in Y_{i, \sigma^!}$. Thus, we can rewrite the entire expression into a single bilinear maximization problem over $\tilde{\bm{\xi}}$ and the multipliers: 
{
\small 
\begin{align*}
    \min_{\plansub}
    \max_{
    i, \lambda,y
    }
    \left \{ 
    \begin{array}{c}
        \sum\limits_{i, \sigma^! \in \check{\Sigma}_{i, j}} 
        \left[ \plansub^T R_{\sigma^!}^\delta z^\delta_{i, \sigma^!} + \plansub^T \left( \lambda^{\delta}_{i, \sigma^!} b^\delta_{i, \sigma^!} \right) \right] + \\
        \sum\limits_{\substack{i, (I, \sigma^!, \cdot) \\ \in \invbreve{\mathcal{B}}_{i, j}}} 
        \left[ \plansub^T R_{I, \sigma^!}^\nu z^\nu_{i, \sigma^!} + \plansub^T \left( \lambda^{\nu}_{i, \sigma^!} b^{\nu}_{i, \sigma^!} \right) \right] + \\
        \sum\limits_{i, (I, \cdot) \in \breve{B}_j} \plansub^T \left( \lambda^v_{i, I} b^v_{i, I} \right)
    \end{array}
    \right \}, 
\end{align*}
}
where $z^\delta_{i, \sigma^!} = \lambda^{\delta}_{i, \sigma^!} \check{y}_{i, \sigma^!}$ and $z^\delta_{i, \sigma^!} = \lambda^{\nu}_{i, \sigma^!} \hat{y}_{i, \sigma^!}$, for appropriately chosen constants $R, b$ (which may vary on $\blueprint$). Hence, we can treat the refinement problem as a \textit{zero-sum} game between a \textit{mediator}, who chooses a refinement $\plansub$ and \textit{deviator}, who chooses multipliers and best-deviating responses. This zero-sum game can be solved by running self-play between two Hannan-consistent regret minimizers and taking average strategies. A regret minimizer for the deviator can be constructed efficiently using existing techniques \citep{farina2019regret} or simply CFR \cite{zinkevich2007regret}. 

We now briefly describe what $R$ and $b$ contain. $R^\delta$ and $R^{\nu}$ are constants for each trigger sequence $\sigma^!$, such that for a best response (weighted by $\lambda^\delta$ and $\lambda^\nu$), $\plansub^T R^\delta z^\delta_{i, \sigma^!}$ gives the largest possible reward for deviating, and in the case of $\nu$ the value of the head infoset. $b$ contains two components (i) the bound (either upper/lower or safety bounds for in-subgame deviations), and (ii) the value of sequences /infosets assuming best-responses.

\section{Decomposition of $\polytopesub$ using scaled extensions} 
We first briefly describe the decomposition algorithm of \citet{farina2019efficient}. The reader is directed there for more details.

\paragraph{Scaled Extensions} The scaled extension operator is a convexity preserving operation between sets. It was used by \citet{farina2019efficient} to incrementally extend the strategy space $\Xi$ in a top-down fashion. 

\begin{definition}{Scaled Extension (\cite{farina2019efficient})}
Let $\mathcal{X}$ and $\mathcal{Y}$ be non-empty, compact, and convex sets, and let $h: \mathcal{X} \rightarrow \mathbb{R}_+$ be a nonnegative affine real function. The scaled extension of $\mathcal{X}$ with $\mathcal{Y}$ via $h$ is defined as the set
\begin{align*}
\mathcal{X} \lhd_h \mathcal{Y} := \{ (\bm{x},\bm{y}): \bm{x} \in \mathcal{X}, y \in h(\bm{x})\mathcal{Y}\}
\end{align*}
\end{definition}
It was shown that scaled extensions preserve convexity, non-emptiness, and compactness of sets. 

\paragraph{Expression $\Xi$ using scaled extensions.} The idea is that some of the structural constraints of $\Xi$ were redundent, and can be expressed in terms of others. For example, supposed we were trying to express $\Xi$ in Figure~\ref{fig:signalling}. In top-down fashion, we begin with $\xi[\varnothing, \varnothing]$. Now, we look at the constraints that sum to $\xi[\varnothing, \varnothing]$. This gives 3 options: expand block 2, or blocks 11 or 12. Let us expand block 2. By `expand' what we do is to apply the scaled extension using a set $\mathcal{Y}$ which is a simplex of size 2. The function $h$ is chosen to be $0$ everywhere except for the index which we are expanding from (in this case $\xi[\varnothing, \varnothing]$. T increases the `size' of the set by 2 dimensions. We can see that the scaled extension have handled the structural constraints so far. It turns out we can repeatedly apply this and fill in blocks 3-8 the same way, and in that order. Now, we need to fill-in blocks 9-12. It turns out, however, that the constraints which involve blocks 9-12 already explicitly fix those entries. That is, there is no longer any degree of freedom for those entries: they can be inferred from the other entries. Crucially, note that block 9 is uniquely determined by blocks 4 and 7, without any inconsistencies or clashes. Since this is the case, we will express 9-12 as a sum of entries in the partially built $\Xi$.  Again, this can be done using the scaled extension, by setting $\mathcal{Y}$ to be a singleton, but choosing $h=0$ except for the indices we want to sum \textit{from}.

The order of expansion in Figure~\ref{fig:signalling} was carefully chosen. For example, if we had expanded 11 and 12 first, followed by blocks 9-10. However, in almost all cases, we have painted ourselves into the corner: block 2's constraints are such that it needs to be summed to by both 9 and 10. For example we set columns $r_x, r_y =0$, and expanded $[\varnothing, \ell_x]$, $[\varnothing, \ell_y] = 1$. Next, we expanded downwards we set $[G, \ell_x]= 1$ and $[B\, ell_y] = 1$. The, we try to `backfill' 2 using the constraints that involve 2. This would end up as having $[G, \varnothing]$ and $[B, \varnothing]$ both being equal to 1. But this  in turn means that they sum to 1 (the structural constraints of $[\varnothing, \varnothing]$). This shows that the order of expansion is crucial: not all will work well. In this example, the choice of expanding block 2 rather than blocks 11, and 12 was crucial. It turns out that a `good' order of decomposition always exists in games without chance.

\begin{definition}{(Critical infosets) \citet{farina2019efficient}}
Let $(\sigma_1, \sigma_2)$ be a relevant sequence pair and let $I_1 \in \mathcal{I}_1$ be an infoset for $P_1$ such that $\sigma(I_1) = \sigma_1$. Inofset $I_1$ is called critical for $\sigma_2$ if there exists at least one $I_2 \in \mathcal{I}_2$ with $\sigma(I_2) = \sigma_2$ such that $I_1 \connected I_2$. A symmetric definition holds for $I_2 \in \mathcal{I}_2$.
\end{definition}
A key result of \cite{farina2019efficient} was that in games without chance, for any relevant sequence pair $(\sigma_1, \sigma_2)$ at least one player has at most one critical infoset for the opponent's sequence. That player is called the \textit{critical player}.

\paragraph{The decompose subroutine}
The \textsc{Decompose} subroutine expands a sequence pair $(\sigma_1, \sigma_2)$ and adds it into $\mathcal{X}$. It is recursive in nature, and proceeds in 2 main steps. See \citet{farina2019efficient} for a much more detailed explanation.

Step 1: Find a critical player from $(\sigma_1, \sigma_2)$, where $\sigma_i = (I_i, a_i)$. This is guaranteed to exist. WLOG let that player be $P_1$. Then, expand all infosets $I$ if $\sigma_2=\varnothing$ or $I_1 \connected (I_2, a_2)$, and $\sigma(I) = \sigma_1$. Each expansion is done using the scaled extension, with $h$ being $0$ everywhere and a $1$ in the index of admission. 

Step 2: For each sequence $\sigma'$ immediately under $I$, call $\textsc{Decompose}(\sigma', \sigma_2)$. After this step, all indices in $\{ (\sigma_1, \sigma_2' ) | \sigma_2' \succ \sigma_2 \}$.

Step 3: We perform backfilling of structural constraints $\{ (\sigma_1, \sigma_2' | \sigma_2' \succ \sigma_2 \}$. First, if the critical player does have a critical infoset, then we will backfill---there are no longer any degrees of freedom. We will assign the value based on the constraint. If there is no critical infoset, then this we split by attaching more scaled extensions to simplexes. 

\paragraph{A key tiebreaking rule.} 
We perform the decomposition of $\Xi$ in the same order one would if we were performing for the full game \cite{farina2019efficient}, except that we terminate whenever encountering a sequence pair $(\sigma_1, \sigma_2) \not\in \hat{\game} \cup \check{\game}_j$. We show that this process expresses $\polytopesub$ and can be viewed as a series of scaled extensions. Consequently, there exists a regret minimizer over $\polytopesub$. 

However, there is an additional tiebreaking element which we may encounter in the decomposition algorithm: when faced with a choice to expand a sequence pair, we should prioritize sequence pairs which do \textit{not} lie in subgames over those which do. This prioritization is natural, since we do not want to fill in sequence pairs in the subgame before pre-subgame sequence pairs. This tiebreaking rule will not interfere with the rest of the decomposition algorithm.

We illustrate the problem with a simple example: a 2x2 matrix game (say, the Chicken Game, since this has interesting equilibria). Here, actions are simply sequences. We call the sequences for $P_i$, $\sigma_{i,1}$ and $\sigma_{i,2}$. Now the matrix game can be expressed in terms of an EFG. We assume player moves first, and takes an action. After that, $P_2$ takes his action, but without knowledge of $P_1$'s action. This is achieved using an information set that spans over the 2 states after $P_1$ took his action. Now, let us suppose a single subgame which contains only $P_2$'s actions, i.e., $P_1$ making his move was pre-subgame. This in turn means that the relevant sequence pairs $(\sigma_{1,1}, \varnothing)$ and $(\sigma_{1,2}, \varnothing)$ are not in a subgame (i.e., in a refinement, they are supposed to follow the blueprint), while all other sequence pairs (except for $(\varnothing, \varnothing)$) are in the subgame.
 
Now, if were to just apply the expansion order of \citet{farina2019efficient}, we would have two options: either expand $\sigma_{1,1}$ and $\sigma_{1,2}$ first, or there other way round. The reader may verify that expanding $\sigma_{1,1}$ and $\sigma_{1,2}$ first works fine. That is, we can expand $\sigma_{1,1}$ and $\sigma_{1,2}$, then $(\sigma_1, \sigma_2)$ (where neither $\sigma_1$ and $\sigma_2$ are the empty sequence). Then, we can backfill $\sigma_{2,1}$ and $\sigma_{2,2}$.

However, if we were to expand alongside $P_2$, i.e., $\sigma_{2,1}$ and $\sigma_{2,2}$, then a nasty situation occurs. After filling up the correlated non-empty sequence pairs, we have to backfill $\sigma_{1,1}$ and $\sigma_{1,2}$. This is not possible, since $\sigma_{1,1}$ and $\sigma_{1,2}$ were part of the blueprint! In general, we are not permitted to perform the backfill when the `source' is in a subgame and the `destination' is pre-subgame. The solution in this example is simple: when there is a tie as to which sequence pairs should be expanded, select the one that is pre-subgame. Can this simple solution always work? If we expanded the infosets that are not in subgames first before the other player's infoset (which is in a subgame), then we can be sure that in the backfill step we will \textit{never} fill in a sequence pair in $\hat{\game}$ with sum of sequence pairs in $\check{\game}_j$. However, can this expansion rule ever conflict with the expansion rule of \cite{farina2019efficient}---i.e., expanding infosets belonging to the critical player? It turns out this clash will never occur.

\begin{theorem}
Let $(\sigma_1, \sigma_2)$ be a relevant sequence pair, and $I_1$, $I_2, I_2'$ be infosets such that $\sigma(I_1)=\sigma_1$, $\sigma(I_2)=\sigma_2$, $\sigma(I_2')=\sigma_2$, and that $I_1 \connected I_2$ and $I_1 \connected I_2'$. (Verify that $I_1$ is a critical infoset for $(\sigma_1, \sigma_2)$.) It cannot be that $I_1$ belongs to some (any) subgame but either or both of $I_2$, $I_2'$ do not. 
\end{theorem}
\begin{proof}
First, note that if any of $I_2$ or $I_2'$ lies in a subgame, it must be the same one as $I_1$, there exists a path starting from the root passing through $I_1$ and that infoset ($I_2$ or $I_2'$). Now, suppose WLOG that $I_2$ is not in a subgame but $I_1$ is. We will demonstrate a contradiction.

First, we know that there exists a state $h_2 \in I_2$ and a state $h_1 \in I_1$ where $h_2 \sqsubset h_1$. This is because $I_2 \connected I_1$. It cannot be that $h_1 \sqsubset h_2$, since that would imply that $I_2$ must belong to a a subgame (which contradicts our assumption). Now, since $I_2' \connected I_1$, there exists a state $h_2' \in I_2$ and $h_1' \in I_1$ which are along a path from the root (though this time, we do not know which precedes the other). 

Let $w$ be the lowest-common-ancestor of $h_2$ and $h_2'$. The state $w$ cannot be a chance node (since $\game$ has no chance). It also cannot belong to $P_1$, since this would mean that $h_1$ and $h_1'$ are in different infosets (they must have been the result of different actions of $P_1$ at $w$). Hence, $w$ must belong to $P_2$. But that is also impossible, since we assumed from the beginning that $\sigma(I_2)=\sigma(I_2')=\sigma_2$. Clearly this cannot be the case since they have different preceding sequences starting from $w$.  
\end{proof}

Put together, this means that the rules of expansion for will not conflict. If there ever needs to be backfilling, it will be either entirely within or outside a subgame, or backfilling entries within a subgame from entries before a subgame, and not the other way around (which will violate the equality-to-blueprint constraints).



\end{document}